\documentclass{elsarticle}
\usepackage{color,amsthm,amsmath,amsxtra,amsfonts,dsfont,graphicx,bm, bbm}

\newcommand{\Id}{\mathbbm{1}}
\newcommand{\be}{\begin{equation}}
\newcommand{\ee}{\end{equation}}
\newcommand{\bea}{\begin{eqnarray}}
\newcommand{\eea}{\end{eqnarray}}

\newtheorem{thm}{Theorem}[section]
\newtheorem*{thm*}{Theorem}
\newtheorem{prop}[thm]{Proposition}
\newtheorem*{prop*}{Proposition}
\newtheorem{cor}[thm]{Corollary}
\newtheorem*{cor*}{Corollary}
\newtheorem{lem}[thm]{Lemma}
\newtheorem{defn}[thm]{Definition}

\newtheorem{ex}[thm]{Example}

\newcommand{\ket}[1]{\vert#1\rangle}
\newcommand{\bra}[1]{\langle#1\vert}
\newcommand{\tr}{\mathrm{tr}}

\begin{document}

\title{Matrix Product Density Operators: Renormalization Fixed Points and Boundary Theories}

\author{J. I. {Cirac}$^1$, D. {P\'erez-Garc\'{\i}a}$^{2,3}$, N. {Schuch}$^1$ and F. {Verstraete}$^{4,5}$}
\address{$^1$ Max-Planck-Institut f{\"{u}}r Quantenoptik,
Hans-Kopfermann-Str.\ 1, D-85748 Garching, Germany}
\address{$^2$ Departamento de An\'alisis Matem\'atico, Universidad Complutense de Madrid, Plaza de Ciencias 3,
28040 Madrid, Spain}
\address{$^3$ ICMAT, Nicolas Cabrera, Campus de Cantoblanco, 28049 Madrid, Spain}
\address{$^4$ Department of Physics and Astronomy, Ghent University}
\address{$^5$ Vienna Center for Quantum Technology, University of Vienna}

\begin{abstract}
We consider the tensors generating matrix product states and density operators  in a spin chain. For pure states, we revise the renormalization procedure introduced in [F. Verstraete et al., Phys. Rev. Lett. {\bf 94}, 140601 (2005)] and characterize the tensors corresponding to the fixed points. We relate them to the states possessing zero correlation length, saturation of the area law, as well as to those which generate ground states of local and commuting Hamiltonians. For mixed states, we introduce the concept of renormalization fixed points and characterize the corresponding tensors. We also relate them to concepts like finite correlation length, saturation of the area law, as well as to those which generate Gibbs states of local and commuting Hamiltonians. One of the main result of this work is that the resulting fixed points can be associated to the boundary theories of two-dimensional topological states, through the bulk-boundary correspondence introduced in [J. I. Cirac et al., Phys. Rev. B {\bf 83}, 245134 (2011)].
\end{abstract}

\maketitle

\section{Introduction}

Tensor networks offer concise descriptions of certain many-body
quantum states in lattices. They thus consitute the basis of several numerical
algorithms to solve strongly correlated systems problems \cite{Valentin,CiracVerstraete},
like the density
matrix renormalization group technique \cite{DMRG,Schollwock}. They have also proven to
be very useful in the description of a variety of phenomena,
including symmetry \cite{Sanz} or gauge invariance \cite{Tagglia,Jutho},
and topologogical \cite{Norberttopo,Franktopo}, and symmetry protected order
\cite{GuWen}.

Matrix product states (MPS), the tensor network states naturally describing one
dimensional systems, are by now relatively well understood and characterized.
In particular, they have given rise to the classification of gapped phases in
one-dimensional systems \cite{Wen,Norbertphases,Pollamann}. A key technique
that has allowed to achieve this goal is the renormalization procedure introduced
in \cite{Latorre} (see also \cite{Wen}). The corresponding fixed points turn out to have a very
simple structure, which indeed leads to the mentioned classification. In some sense
they describe the coarse-grained behavior of all possible MPS, but with the
peculiarity of having zero correlation length, as it should be expected for
any renormalization procedure applied to gapped phases. Renormalization
procedures are also at the realm of the multi-scale entanglement renormalization
states \cite{Mera}, describing critical phenomena.

The extension of tensor network renormalization methods to higher dimensions is
not a simple task. Simple renormalization procedures have been introduced in
\cite{WenRG2D,VidalRG2D}, and the corresponding fixed point sets have been analyzed to
some extent \cite{WenRGFP}. In particular, toric code \cite{Toric} and string net states
\cite{StringNet} belong to those sets. They have zero correlation length and
possess topological order, which has been extensively investigated using
projected entangled-pair states (PEPS) \cite{PEPS1,PEPS2}, the natural generalization of MPS to
higher spatial dimensions.

In this work we first thoroughly revisit the fixed points of the renormalization
procedure introduced in \cite{Latorre} for one-dimensional spins and then extend
the corresponding formalism to mixed states; that is, first we concentrate on MPS and
then on matrix product density operators (MPDO). In this context, we define concepts
like zero correlation length, saturation of the area law, as well as local
commuting Hamiltonians, and connect them to the characterization of the
fixed points. The main result of this work is the characterization of
what we call renormalization fixed points (RFP) for MPDO. As we will see, they
give rise to an algebraic structure that directly connects with fusion categories and
 string net states \cite{StringNet,FrankStrings}. This is, indeed, what one
could expect given the existing bulk-boundary correspondence in PEPS, if we
describe the boundary theory in terms of a MPDO \cite{bulkboundary}. Thus, some of the
RFP can be interpreted as the boundaries of string net states. In this sense,
the MPDO we find can be understood \cite{bulkboundconj} as the product of a Gibbs state of a local
commuting Hamiltonian with a non-local projector (which commutes with the
Gibbs state), and which reflects an anomaly
\cite{anomaly}. This can be interpreted as the consequence of the topological order
of a two-dimensional theory at the boundary.

A basic tool of the characterization of the RFP, both for pure and mixed states, is
the so-called canonical form of MPS \cite{David2006}. In this paper we also carefully
revise this notion, and derive some other results that are required for the
formal proofs of our statements. In fact, since we introduce many definitions and
such statements, and the corresponding proofs are sometimes involved, we have
structured the paper as follows. First we introduce the main definitions and
theorems, and we leave the technical material for the appendices. The paper
starts out with the definition and characterization of canonical forms that include
both MPS and MPDO, and in the next two sections we concentrate on the RFP of MPS and MPDO,
respectively.

\section{Matrix Product Vectors}\label{Sec:MPV}

\subsection{Definition}

The main object of study here is a tensor, $A$, with coefficients $A^i_{\alpha,\beta}$, where $i=1,\ldots,d$ and $\alpha,\beta=1,\ldots,D$. We denote by $H_d$ a Hilbert space of finite dimension $d$, and by $\{|i\rangle\}_{i=1}^d$ an orthonormal basis in $H_d$. The tensor $A$ defines a family of vectors:

\begin{defn}
The vectors
 \be
 \label{MPV}
 |V^{(N)}(A)\rangle = \sum_{i_1,\ldots,i_N=1}^d {\rm tr}\left( A^{i_1}\ldots  A^{i_N}\right) |i_1\rangle\otimes\ldots\otimes|i_N\rangle \in H_d^{\otimes N},
 \ee
for $N=1,2,3,\ldots$ are the {\em matrix product vectors generated} by $A$. We call $D,d<\infty$ the {\em bond dimension} and the {\em physical dimension}, respectively.
\end{defn}

We will omit the tensor products, and thus write $|i_1,i_2\rangle$ instead of $|i_1\rangle\otimes|i_2\rangle$. Sometimes, we will use round kets and Greek letters, eg. $|\alpha)$, to denote the elements of an orthonormal basis of vectors acting on the indices that are contracted in order to generate the MPV.

For the moment, we do not specify whether each MPV represents a state, in which case we will have a matrix product state (MPS), or an operator, in which case we will have a matrix product operator (MPO).We will denote by ${\cal V}(A)$ the family of MPV generated by $A$, i.e.
 \be
 {\cal V}(A) = \{|V^{(N)}(A)\rangle\}_{N\in\mathds{N}}.
 \ee
The map
 \be
 A \to {\cal V}(A)
 \ee
is not one-to-one. In fact, different tensors may generate the same family of vectors, i.e. $|V^{(N)}(A)\rangle = |V^{(N)}(B)\rangle$. In the following sections we will analyze the relations among the tensors $A$ generating the same MPV's. In particular, we will show that one can always express a tensor in a canonical form so that one can easily compare the MPV's generated by different tensors in that form. This will be a key element of our analysis in the following sections.

By definition, the vector $V^{(N)}(A)$ is translationally invariant. One can also consider different matrices corresponding to different subsystems. However, here we will always deal with MPV of this form, and thus drop the word translationally invariant. Apart from that, in order not to overload the notation, we will drop the index $N$ and/or the dependence on $A$ in $V^{(N)}(A)$ and denote it simply by $V$ whenever this causes no confusion (or the corresponding family by ${\cal V}$). We will denote by $A^i$ the $D\times D$ matrices whose matrix elements are precisely the coefficients $A^i_{\alpha,\beta}$ of the tensor $A$, and by ${\cal M}_{D\times D}$ the whole set of $D\times D$ matrices.


\subsection{Graphical representations}\label{section:graphical}

Since we will deal with various tensors contracted in different forms, the notation may become cumbersome. Thus, it will be useful to use a graphical representation in certain occasions. We will write a tensor in terms of a box with lines representing the different indices. For instance, the tensor $A$ generating MPV will be represented by
 \be
 A =  \raisebox{-12pt}{\includegraphics[height=2.8em]{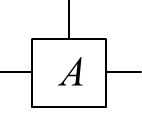}}
 \ee
where the line pointing up represents the physical index, and the other two the auxiliary indices. The contraction of two tensors will be represented by joining the lines of the corresponding indices.
For instance, the coefficients of the MPV (\ref{MPV}) in the basis $|i_1,\ldots,i_N\rangle$ will be represented as
 \be
 \langle i_1,\ldots,i_N|V\rangle = \; \raisebox{-8pt}{\includegraphics[height=3.8em]{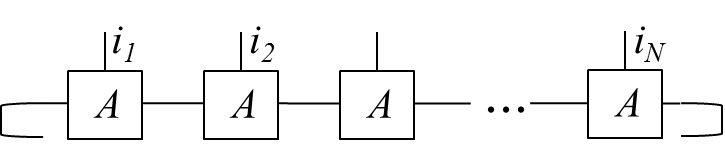}}
 \ee
The meaning of this expression is that each coefficient can be written by multiplying the matrix $A^{i_N}$ by $A^{i_{N-1}}$, then by $A^{i_{N-2}}$, and so on. The curly lines at the end indicate that we have to take the trace at the end (i.e., contract the first with the last index after the matrix multiplication). We will abuse notation and simply write
 \be
 |V\rangle = \; \raisebox{-10pt}{\includegraphics[height=2.8em]{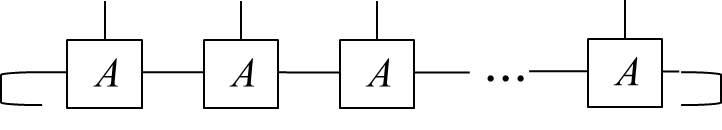}}
 \ee
without explicitly writing the indices $i_n$.

Furthermore, since we will have to represent relative complex expressions with tensors, some of the lines will have to cross. If we write a full circle in the crossing, this will mean a delta function. For instance
 \be
\raisebox{-12pt}{\includegraphics[height=4.2em]{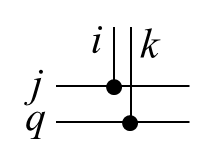}}
 \ee
means that the indices on the right and the left coincide, and that $i=j$ and $k=q$ (but not that $i=k$).


\subsection{Canonical forms}

Two tensors, $B$ and $C$, generate the same MPV if, for instance, all the matrices $B^i$ and $C^i$ are related by a similarity transformation
  \be
  B^i=X C^i X^{-1}
  \ee
where $X$ is a nonsingular matrix. Graphically,
 \be
 {\includegraphics[height=2.8em]{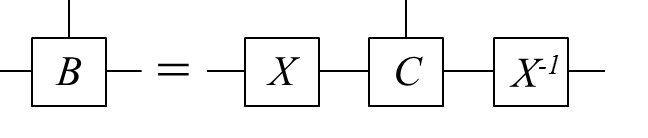}}
 \ee
We say that the tensors $B$ and $C$ are related by the {\em gauge matrix} $X$, and that the MPV are invariant under the similarity (gauge) transformation induced by $X$. But this is not the only way two tensors can generate the same MPV. It can also happen that all $B^i$ have a block upper triangular form, eg.
 \be
 B^i = \left(\begin{array}{cc} B^i_1 & B^i_o\\ 0 & B^i_2 \end{array}\right)
 \ee
where $B^i_k$ are $D_k\times D_k$ matrices, and $B^i_o$ is a $D_1\times D_2$ matrix. The MPV generated by $B$ does not depend on $B^i_o$, so that for any such matrix we will have the same vector. In fact, there exists a subspace, $S_1$, of dimension $D_1$ which is left invariant under the action of all $B^i$, i.e. $B^i S_1\in S_1$. In order to lift this redundancy, it is convenient to select out of all the matrices $B^i$ generating the same MPV, those which have a simplified form. For instance, we can choose $B^i_o=0$, and use the invariance under the gauge transformation induced by a gauge matrix to specify certain properties of the $B^i$. In the following we will systematize this procedure to define canonical forms of MPV's.

First of all, let us assume that there exist one or more subspaces that are invariant under the action of all $B^i$. Let us denote by $S_1$ one such subspace, of dimension $D_1$, which does not contain any other invariant subspace, and let us denote by $P_1$ ($Q_1=\Id-P_1$) the orthogonal projector onto $S_1$ ($S_1^\perp$). We have
 \be
 B^i P_1 = P_1 B^i P_1,\quad  Q_1 B^i = Q_1 B^i Q_1.
 \ee
We consider now $Q_1 B^i Q_1$ and proceed in the same form: identify an invariant subspace $S_2\subseteq S_1^\perp$ that does not contain any invariant subspace, denote by $P_2$ the corresponding projector, by $D_2$ the dimension of $S_2$, and $Q_2=\Id-P_1-P_2$. In general, denoting by
 \be
 Q_l = \Id - \sum_{k=1}^l P_k,
 \ee
we have that
 \be
 Q_l B^i P_{l+1} = P_{l+1} B^i P_{l+1}, \quad
 Q_l B^i = Q_l B^i Q_l.
\ee
After a finite number of steps, $r$, there will be no (non-trivial) invariant subspace anymore. The matrices $\{A^i\}$, defined as follows
 \be
 \label{eq:II_Aiplusk1}
 A^i = \sum_{k=1}^r P_k B^i P_k = \oplus_{k=1}^r \mu_k A^i_k,
 \ee
generate the same family of MPV, $\cal{V}$. Here, $r$ is the number of blocks, the $\mu_k$ are complex numbers, and $A^i_k$ are $D_k\times D_k$ matrices with $\sum_k D_k\le D$ (i.e., there can be zero blocks). We associate each tensor $A_k$ a completely positive map (CPM), ${\cal E}_k$, defined through
 \be
 \label{Ek}
 {\cal E}_k (X) = \sum_{i=1}^d A_k^i X A_k^{i\dagger}.
 \ee
The values of $\mu_k$ are chosen such all these maps
have spectral radius equal to 1.

As it was shown in \cite{David2006}, each CPM ${\cal E}_k$ has a unique eigenvalue $\lambda=1$ . However, there may be other eigenvalues of magnitude one, $e^{i 2\pi q/p}$, with $p,q$ integers, gcd$(q,p)=1$, and where $p$ is a divisor of $D$. They correspond to so-called $p$-periodic vectors \cite{Fan92}. In this context, it is very convenient to define the procedure of blocking, which will also play a role when we talk about renormalization procedures. By blocking tensors $B$ we mean defining a tensor $C$ such that the corresponding matrices are obtained by the $p$-fold products of the original ones, i.e. $C^i=B^{i_1}\ldots B^{i_p}$, where the index $i$ is contains all the $i_k$, $k=1,\ldots,p$. Graphically,
 \be
 \raisebox{-12pt}{\includegraphics[height=3.5em]{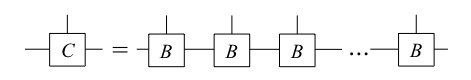}}
 \ee
If there are $t$ $p$-periodic vectors, with $p=p_1,\ldots,p_t$, by blocking lcm$(p_1,\ldots,p_t)$ systems, one obtains a vector without $p$-periodic ones. Since in the present work we will be mostly interested in MPVs obtained after a renormalization procedure which will block large numbers of spins anyway, we will be able to assume that for each $\tilde {\cal E}_k$ there is a unique eigenvalue of magnitude (and value) equal to one. Note that, when discussing MPOs, we may not be allowed to block tensors in some special ways, so that in that case we will have to ensure that there are no $p$-periodic vectors in case we want to use the canonical forms defined below.

\begin{defn}
We call a tensor, $A$, a {\em normal tensor} (NT) if: (i) there exist no non-trivial projector $P$ such that $A^iP=PA^iP$; (ii) its associated CPM, has a unique eigenvalue of magnitude (and value) equal to its spectral radius, which is equal to one. We call $|V\rangle$ {\em normal matrix product vector} (NMPV) if it is generated by a normal tensor.
\end{defn}

\begin{defn}
We say that a tensor,$A$, is in a {\em canonical form} (CF), if
 \be
 \label{eq:II_CF1}
 A^i = \oplus_{k=1}^r \mu_k A^i_k,
 \ee
and the tensors  $A_k$ are NT.
\end{defn}

Note that we are not considering here normalized states, so that we can always choose $|\mu_k|\le 1$ and at least one of them equals one, something which we will assume from now on.
The next propositions follow from the above procedure.

\begin{prop}
After blocking, for any tensor, $B$, it is always possible to obtain another one, $A$, in CF and generating the same ${\cal V}$.
\end{prop}

\begin{prop}
A tensor, $A$, is in CF if it has no $p$-periodic vectors and for every projector, $P$, fulfilling $PA^i=PA^iP$, we also have $A^iP=PA^iP$.
\end{prop}


The CF of a tensor (\ref{eq:II_CF1}) immediately implies that
 \be
 \label{eq:II_Psi_k}
 |V^{(N)}(A)\rangle = \sum_{k=1}^r \mu_k^N|V^{(N)}(A_k)\rangle,
 \ee
Note that two or more families ${\cal V}(A_j)$ may be the same, or related by a phase. This happens, for instance, if two of the tensors, say $A_j$ and $\tilde A_k$, satisfy
 \be
 \label{eq:II:A=XAX}
 \tilde A_k = e^{i\phi_k} X_k A_j X_k^{-1}
 \ee
for some invertible $X_k$ and phase $\phi_k$. This case has to be treated with some care, and for that we will introduce the following

\begin{defn}
The tensors $A_j$ ($j=1,\ldots,g$) form a {\em basis of normal tensors} (BNT) of a tensor $A$ if:
(i) the $A_j$ are NT; (ii) for each $N$, $V^{(N)}(A)$ can be written as a linear combination of $V^{(N)}(A_j)$; (iii) there exists some $N_0$ such that for all $N>N_0$, $V^{(N)}(A_j)$ are linearly independent.
\end{defn}

Note that for any tensor, there always exists a BNT. In \ref{AppendixMPV} we show how to do that in the proof of the following characterization of BNT.

\begin{prop}\label{prop:char-BNT}
The tensors $A_j$ ($j=1,\ldots,g$) form a BNT of $A$ iff: (i) for all NT appearing in its CF (\ref{eq:II_CF1}), $\tilde A_k$, there exists a $j$, a non-singular matrix, $X_k$, and some phase, $\phi_k$, such that (\ref{eq:II:A=XAX}) holds; (ii) the set is minimal, in the sense that for any element $A_j$, there is no other $j'$ for which (\ref{eq:II:A=XAX}) is possible.
\end{prop}

We can write now the matrices of any tensor, $A$, in CF in terms of a BNT, $A_j$, as
 \begin{subequations}
 \bea
 \label{eq:II_ABasicTensors}
 A^i &=& X \left[\oplus_{j=1}^g \left(M_j \otimes A^i_j\right)\right] X^{-1} \\
 &=&
 \oplus_{j=1}^g \oplus_{q=1}^{r_j} \mu_{j,q} X_{j,q} A^i_j X_{j,q}^{-1},
 \eea
 \end{subequations}
where $M_j$ is a diagonal matrix with coefficients $\mu_{j,q}$,
and
 \be
 \label{eq:II_X}
 X= \oplus_{j=1}^g \oplus_{q=1}^{r_j} X_{j,q},
 \ee
so that
 \be
 \label{decBSV}
 |V^{N}(A)\rangle = \sum_{j=1}^g \left(\sum_{q=1}^{r_j} \mu_{j,q}^N \right) |V^{(N)}(A_j)\rangle.
 \ee

We will express (\ref{eq:II_ABasicTensors}) graphically as follows
 \be
 \label{Eq19}
  \raisebox{-12pt}{\includegraphics[height=7.3em]{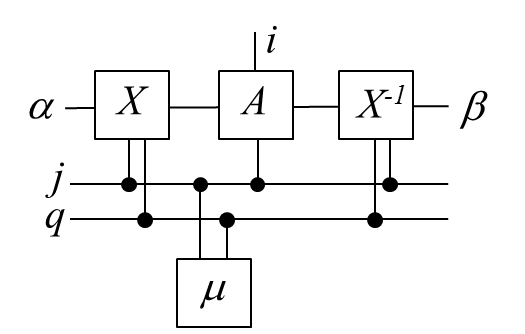}}
 \ee
where
 \be
  A_j^i=\raisebox{-28pt}{\includegraphics[height=5.6em]{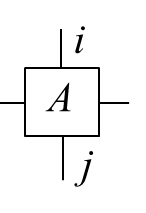}}
 \ee
are the BNT. We have included in (\ref{Eq19}) all the indices, in order to make the relation with (\ref{eq:II_ABasicTensors}) more evident. Note that we have represented the indices $j$ and $q$ by lines, and thus replaced the direct sums by tensor products. In reality, the values over which the indices $q,\alpha$, and $\beta$ run will generically depend on $j$; for instance, $q=1,\ldots,r_j$. Thus, the graphical representation should not be taken literally.

Now we introduce a basic definition in the context of MPS, on which many of the special properties of those states rely.

\begin{defn}
\label{defnbi}
A NT $A$ is called {\em injective}, if the matrices $A^i$ span the whole set of ${\cal M}_{D\times D}$ matrices.
We say that a tensor $A$ is in {\em block injective} canonical form (biCF) if it is in CF, and for each element $X\in \oplus_{j=1}^g {\cal M}_{D_j\times D_j}$ there exists a vector, $c(X)$, such that
$X=\sum_i c_i(X) \tilde A^i$, where $\tilde A^i:=\oplus_{i=1}^g A^i_j$, and $A_j$ are a BNT of $A$.
\end{defn}

The first definition is equivalent to the existence of another tensor, $A^{-1}$, such that $\sum_{i=1}^d (A^i)_{\alpha,\beta} (A^{-1,i})_{\alpha',\beta'}=\delta_{\alpha,\alpha'}\delta_{\beta,\beta'}$, and the second one to the existence of tensors $A_j^{-1}$ such that 
$\sum_{i=1}^d (A^i_j)_{\alpha,\beta}
(A^{-1,i}_{j'})_{\alpha',\beta'}=\delta_{\alpha,\alpha'}\delta_{\beta,\beta'}\delta_{j,j'}$; graphically,
 \be
 \label{injectivity}
 {\includegraphics[height=6em]{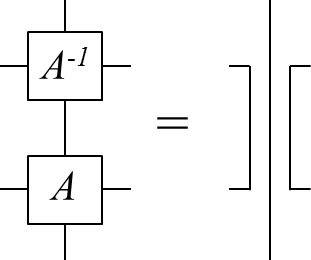}}
 \ee

According to the quantum version of Wieland's theorem \cite{Wieland}, after blocking at most $D^4$ times, every NT becomes injective. We will use this property several times in the coming sections. The definition of biCF basically says that acting on the physical index we have access to each and every element of every BNT. 


Furthermore, in \cite{David2006} it was shown that after $3(L+1)(D-1)$ blockings, any tensor $A$ becomes block injective,  where $L$ is the number of blockings after which each NT becomes injective. Thus we have

\begin{prop}
\label{propblockinj}
After blocking at most $3D^5$ spins, any tensor $A$ in CF is in  biCF.
\end{prop}

We are then ready to state the main result of this section, that we call {\em The Fundamental Theorem of Matrix Product Vectors}. It clarifies which is the freedom for two tensors $A$ and $B$ in CF, to generate families of MPV, ${\cal V}_{a,b}$, such that for each $N$ they are proportional or equal to each other

\begin{thm}
\label{thm1}
Let $A$ and $B$ be two tensors in CF, with BNT $A^i_{k_a}$ and $B^i_{k_b}$ ($k_{a,b}=1,\dots,g_{a,b}$), respectively. If for all $N$, $A$ and $B$ generate MPV that are proportional to each other, then: (i) $g_a=g_b=:g$; (ii) for all $k$ there exists a $j_k$, phases $\phi_{k}$, and non-singular matrices, $X_{k}$ such that $B^i_{k}=e^{i\phi_{k}}X_{k}A^i_{j_k}X_{k}^{-1}$.
\end{thm}

\begin{cor}
\label{II_cor2}
If two tensors $A$ and $B$ in CF generate the same MPV for all $N$ then: (i) the dimensions of the matrices $A^i$ and $B^i$ coincide; (ii) there exists an invertible matrix, $X$, such that
 \be
 \label{eq:II_auxcor}
 A^i = X B^i X^{-1}.
 \ee
\end{cor}

The justification of the name came from the wide and important uses of that result. Besides being the key technical tool in this paper, it has been also the key result in the characterization of symmetries and string order parameters of MPS \cite{Sanz}, the classification of phases in 1D systems \cite{Wen,Norbertphases}, or the construction of string nets with PEPS \cite{FrankStrings}.

\section{Pure States: Renormalization of Matrix Product States}\label{Sec:MPS}

Here, we consider $N$ spins in a 1D chain, and a translationally invariant MPS. That is, now $H_d$ is the Hilbert space for each spin, and thus the vectors $V^{(N)}(A)$ represent states of those spins.
We are interested in defining a renormalization procedure whose fixed points will be the central subject of this section. The renormalization flow will encompass the idea of coarse-graining transformations in space, as it appears in some versions of the standard renormalization group methods. However, it is defined in the space of tensors generating the MPS, as it should affect the whole family of states generated by the tensors. We will analyze the fixed points of that flow, ie. the renormalization fixed points (RFP), as well as their connection with other concepts that naturally appear in the context of many-body systems, like zero correlation length (ZCL), or ground states of frustration-free, commuting Hamiltonians (GSCH). We will finally show how all those notions are equivalent in the main theorem of this section. The other main result of this section characterizes being an RFP in the structure of the local tensor $A$.

\subsection{Renormalization Flow and Renormalization Fixed Point}

We consider the renormalization procedure introduced in \cite{Latorre}. There the idea was to separate MPS in classes such that their elements, when blocking spins, were related by local unitary transformations. The blocking of spins in this procedure gives rise to a coarse grained picture of MPS, and defines a flow in the set of tensors generating the MPS. In the limit $m\to \infty$, this flow converges to a special set of tensors, whose elements can be considered as the fixed points of the renormalization procedure.

In renormalization theory one is interested in global features of a system and hence, in order to define the  renormalization flow, we will identify those states that are related by local unitary transformations. That is, tensors $A$ and $\tilde{A}$ are considered equivalent if they are related by  an isometry $U$:
 \be\label{eq:equivalence-relation-RFP-pure}
 {\includegraphics[height=5em]{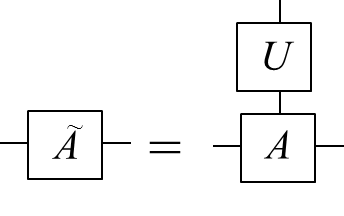}}
 \ee
This divides the set of tensors in equivalence classes.
We will define now the basic step of the renormalization procedure that maps a class of tensors into another class. We will do that by considering how one element of the class is mapped into another element of the new class, as this definition respects the equivalence classes. In particular, let us consider a tensor $B$ and perform the blocking of two spins as defined in the previous section. Thus, we have the new tensor, $C$, with components $C_{\alpha_1,\alpha_2}^{i_1,i_2}=(\alpha_1|B^{i_1} B^{i_2}|\alpha_2)$. This can be viewed as a matrix, $C^i_\alpha$, where $i=(i_1,i_2)$ and $\alpha=(\alpha_1,\alpha_2)$ are combined indices. Since the rank of this matrix, $d_B\le D^2$ (the possible values taken by $\alpha$), we can always find a $d^2\times d_B$  isometry, $U$, with $U^\dagger U=\Id_{d_B}$ and write $C=UA$. That is,
 \be
 \label{RG}
 B^{i_1} B^{i_2} = \sum_{j=1}^{d_B} U_{i,j} A^j
 \ee
for some matrices $A^j$ and $d_B\le D^2$. The flow just maps the class of tensors represented by $B$ in that represented by $A$.

We are interested in the form of the tensor $A$ with matrices $A^j$  that appear as limits of such renormalization flow when $m\to \infty$. By the very definition (\ref{RG}) of the renormalization flow, one clearly expects the following

\begin{thm}\label{thm:renormalization-flow}
A tensor $A$ appears as a limit in the above renormalization flow if and only if
 \be
 \label{AA=A}
 A^{i_1} A^{i_2} = \sum_{i_1,i_2,j} U_{(i_1,i_2),j} A^{j_1}
 \ee
for an isometry $U$.
\end{thm}

Graphically,
 \bea
 &&{\includegraphics[height=5em]{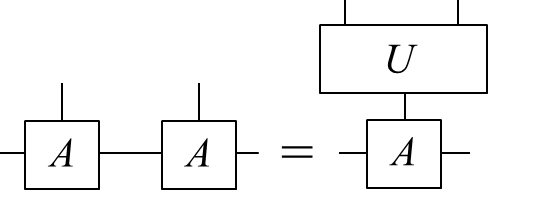}}\\
 &&{\includegraphics[height=5em]{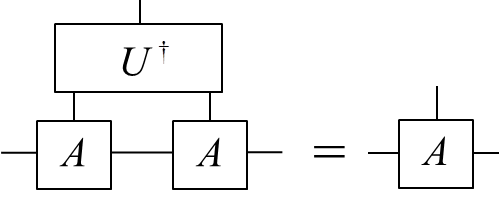}}
 \eea
and, schematically
 \be
 {\includegraphics[height=5em]{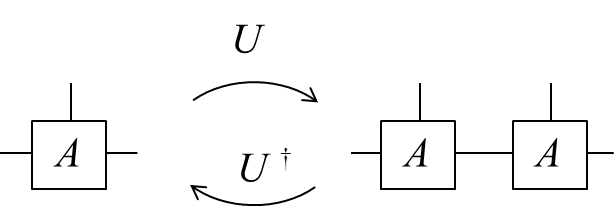}}
 \ee

Moreover, it will be also shown in \ref{AppendixPure} that if we start with a tensor $B$ in CF, the above renormalization flow always converges.
This result motivates the following

\begin{defn}
\label{defRFP}
A tensor $A$ is called a {\em renormalization fixed point} (RFP) if it verifies
(\ref{AA=A}), that is, if one can obtain two tensors out of one (and vice versa) by the action of an isometry in the physical index.
\end{defn}


\subsection{Zero correlation length}

The renormalization flow defined above does not change the state, but just  redefines the spins as it makes them larger and larger by blocking them. Thus, one should expect that if there is any length scale in the system, it will disappear at the end of the flow. In this section we will show how this is captured at the RFP.

\subsubsection{Correlations independent of the distance and local orthogonality}

Perhaps the main notion of distance that naturally appears in the context of the many-body states we are dealing with is the correlation length. This is roughly defined as the typical length at which correlations decay. According to what we expressed above, the RFP should have zero correlation length. This statement can be made more specific as follows. Let us denote by $R_{n_1,n_2}$ the set of spins with $n_1\le n\le n_2$. Then
\begin{defn}
We say that $\Psi$ has {\em correlations independent of the distance} (CID) if for any pair of observables, $O_{n_1,n_1'}$, $O_{n_2,n_2'}$, acting on any disjoint regions $R_{n_1,n_1'}$ and $R_{n_2,n_2'}$, respectively,
 \be
 \label{ZCL}
 \langle \Psi| O_{n_1,n_1'}O_{n_2,n_2'}|\Psi\rangle =
 \langle \Psi| O_{n_1,n_1'}O_{n_2+\Delta,n_2'+\Delta}|\Psi\rangle
 \ee
for all $n_1'-n_2+1 < \Delta < N-n_2'$.
\end{defn}

Even though this definition is probably the most natural one, there are states with CID which keep other type of lengths (and hence cannot be RFP) as the following trivial example shows.

\begin{ex}
Let us consider the state
 \be
 \label{Ex:ZCL}
 |\Psi\rangle = |0,\ldots,0\rangle + |+,\ldots,+\rangle,
 \ee
where $|+\rangle=(|0\rangle+|1\rangle)/\sqrt{2}$. The corresponding matrices are
 \be
 A^0=\left( \begin{array}{cc} 1 & 0 \\ 0 & 1/\sqrt{2} \end{array}\right), \quad A^1=\left( \begin{array}{cc} 0 & 0 \\ 0 & 1/\sqrt{2} \end{array}\right).
 \ee
This state is permutationally invariant, and thus it trivially satisfies (\ref{ZCL}). However, it does not fulfill (\ref{AA=A}), and thus it is not a RFP. In fact, if we block $n$ spins we will have
 \be
 |\Psi\rangle = |\tilde 0,\ldots,\tilde 0\rangle + |\tilde +,\ldots,\tilde +\rangle,
 \ee
where $|\tilde 0\rangle=|0,\ldots,0\rangle$ and $|\tilde +\rangle= |+,\ldots,+\rangle$. Although this state looks very similar to the original one, now the scalar product $\langle \tilde 0|\tilde +\rangle$ decreases exponentially with $n$. Thus, if we think in terms of renormalization, the state (\ref{Ex:ZCL}) cannot be a renormalization fixed point despite the fact that has CID.
\end{ex}

The previous example illustrates that if we have CID and we keep on blocking, the states generated by the different elements of the BNT will look locally orthogonal. This property can be defined as follows
\begin{defn}
\label{DefLO}
We say that two tensors, $A_j, A_{j'}$ are {\em orthogonal} if
 \be
 \sum_i A^i_j \otimes \bar A^i_{j'} = 0.
 \ee
We will say that a MPV is {\em locally orthogonal} (LO) if the associated BNT are orthogonal. \end{defn}

\begin{defn}
We will say that a tensor $A$ has ZCL if the generated family of MPS is LO and has CID.
\end{defn}

The property of ZCL for a tensor can be easily characterized using the so-called transfer matrix.

\begin{defn}
Given a tensor $A$, generating a family of MPV, we define the {\em transfer matrix}
 \be
 \mathbb{E} = \sum_{i} A^i \otimes \bar A^i
 \raisebox{-24pt}{\includegraphics[height=4.8em]{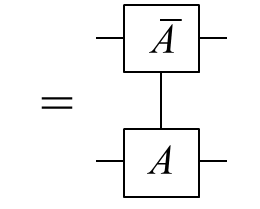}}
 \ee
 \end{defn}

The transfer matrix is the standard tool to calculate expectation values. For any pair of observables, $O_{1,2}$, acting on subsystems $s_{n_{1,2}}$, respectively,
 \be
 \label{Corr}
 \langle V|O_1 O_2|V\rangle = {\rm tr}\left(
 \mathbb{E}_{O_1}\mathbb{E}^{n_2-n_1-1} \mathbb{E}_{O_2} \mathbb{E}^{N+n_1-n_2-1}\right),
 \ee
where $\mathbb{E}_O$ is a matrix that depends on the matrices $A^i$ as well as on the operator $O$.

Note that if $\mathbb{E}^2=\mathbb{E}$, then one will have CID. In fact, this condition is equivalent to ZCL:

\begin{thm}
\label{TheoremZCLPure}
A tensor $A$ in CF has ZCL if and only if $\mathbb{E}^2=\mathbb{E}$.
\end{thm}

\subsection{Commuting parent Hamiltonians}

One can always build local Hamiltonians for which a MPS is a ground state. Among them, the so-called parent Hamiltonians play an important role, as they ensure the ground space is very much related to the structure of the tensors generating the MPS. They are additionally translationally invariant and frustration free; that is, a sum of positive operators that act on a neighborhood of each spin, and that annihilate the state. In this section we will show that if, additionally, those local operators commute with each other, then there is an intimate relation between their ground states and the RFP. In \ref{AppendixLooseEnds} we also derive a general relation between ground states of commuting parent Hamiltonians and the absence of correlations.

Let us consider a tensor $A$ generating a family of MPS. Let us take a block consisting of $L$ spins, and the subspace, $S_L$, spanned by
 \be
 |v_{\alpha,\beta}\rangle = \sum_{i_1,\ldots,i_L} (\alpha|A^{i_1}\ldots A^{i_N}|\beta) |i_1,\ldots,i_N\rangle
 \ee
Let us further assume that $d^L>D^2\ge {\rm dim}(S_L)$, so that $S_L^\perp$ is non-trivial. We denote by $P_L^\perp$ the projector onto that subspace and construct the Hamiltonian
 \be
 \label{eq:parent}
 H_{L}^{(N)} = \sum_{j=0}^{N-1} \tau_j(P_L^\perp)
 \ee
where $\tau_j$ translates the spins by an amount $j$, so that it is translationally invariant (with periodic boundary conditions). Obviously, $|V^{(N)}(A_j)\rangle$ are ground states of $H_L^{(N)}$ with zero energy.

\begin{defn}
Given a tensor in CF, with BNT $A_j$ ($j=1,\ldots,g$), we say that
(\ref{eq:parent}) is a parent Hamiltonian if the ground state subspace is spanned by $|V^{(N)}(A_j)\rangle$ for all $N>L$.  A commuting parent Hamiltonian fulfills $[\tau_j(P_L),P_L]=0$ for all $j=1,\ldots,L-1$. A nearest-neighbor parent Hamiltonian has $L=2$.
\end{defn}

A direct consequence of Proposition \ref{propblockinj} \cite{David2006} is that for any $A$ one can always find a parent Hamiltonian with $L\le 3D^5$ independent of the tensor $A$ (just of its bond dimension). Thus, if we block spins it is always possible to find a nearest-neighbor parent Hamiltonian. In the following section we will establish a close relation between RFP and the ground states of nearest-neighbor commuting parent Hamiltonians (NNCPH).

\subsection{Main theorems}

We can now state the main results of the section. In Theorem \ref{thm:main-MPS}, we show the announced equivalence between being a RFP,  ZCL, and NNCPH. Theorem \ref{thm:charact-MPS} gives a characterization of the tensors $A$ with the RFP property.

\begin{thm}\label{thm:main-MPS}
The following properties are equivalent for a tensor $A$ in CF generating a family of  MPS.
\begin{description}
\item[(i)] RFP.
\item[(ii)] ZCL.
\item[(iii)] For all $N>2$, $|V^{(N)} (A)\rangle$ is a ground state of a NNCPH.
\end{description}
\end{thm}

\begin{thm}\label{thm:charact-MPS}
A tensor $A$ in CF is a RFP iff it can be written as
 \be
 \label{III_CFI_RFP}
 A^i = \oplus_{j=1}^g \oplus_{q=1}^{r_j} \mu_{j,q} X_{j,q} \Lambda_j U^i_j X_{j,q}^{-1}
 \ee
where $|\mu_{j,q}|=1$, the matrices $\Lambda_j$ are diagonal, positive, and $\tr(\Lambda_j)=1$, and $U$ is an isometry in the sense that
 \be
 \label{eq:III_isometry}
 \sum_{i=1}^s (U_j^i)_{\alpha,\beta}  ({\bar U}_{j'}^i)_{\alpha',\beta'} = \delta_{j,j'}
 \delta_{\alpha,\alpha'}\delta_{\beta,\beta'}.
 \ee
\end{thm}

We can represent this graphically as
 \be
 \raisebox{-24pt}{\includegraphics[height=7em]{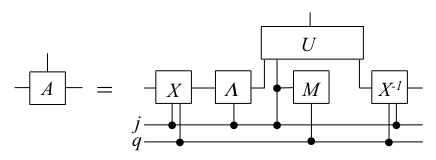}}
 \ee
This expression has the following meaning. There are two indices,
$j,q$, which give rise to the different blocks and are indicated in the graph. The index $j$ enumerates the elements of a BNT, whereas the $q$ gives different representations of the same NT. The tensors $X$, and $M$ depend on the values of those two indices, whereas $\Lambda$ only depends on $j$. The resulting matrices, $X_{j,q}$ and $\Lambda_j$ are multiplied and, together with $X^{-1}_{j,q}$, and the indices $j$ and $q$ give rise to the physical index, after applying the isometry $U$. Note that since the indices $j$ and $q$ connect the physical and virtual indices, up to the isometry, they give rise to a direct sum in both physical and virtual spaces.

The MPS generated by a RFP tensor $A$ in CF thus have the form
 \be
 \label{eq:III_MPSRFP}
 |V^{(N)}(A)\rangle = \sum_{j=1}^g \left(\sum_{q=1}^{r_j} e^{iN\phi_{j,q}}\right) |V^{(N)}(A_j)\rangle,
 \ee
where its basic vectors give
 \be\label{eq:basic-vectors-RFP-pure}
 |V^{(N)}(A_j)\rangle = U^{\otimes N} |\varphi_j\rangle^{\otimes N},
 \ee
with
 \be
 \label{eq:III_varphi}
 |\varphi_j\rangle = \sum_{m_j=1}^{d_j} \lambda_{m_j} |m_j,m_j\rangle.
 \ee
Note that one can understand the state by taking two spins per node, $a_n$ and $b_n$: the state $\varphi_j$ is shared by any two spins in neighboring nodes, $b_{n}$ and $a_{n+1}$, whereas the isometry $U$ acts on the two spins at each node, $a_n$ and $b_n$. The state can be graphically represented as
 \be
  |V^{(N)}(A_j)\rangle = \raisebox{-6pt}{\includegraphics[height=4.0em]{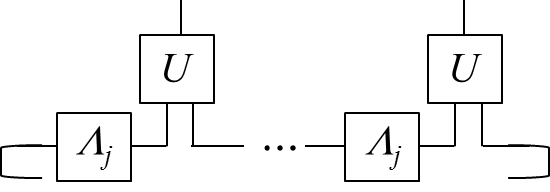}}
 \ee

\begin{cor}
\label{III_cor3}
The elements of a BNT corresponding to a tensor $A$ that is RFP have the form
 \be
 A_j = X_j \Lambda_j U^i_j X_j^{-1}
 \ee
where $U$ is an isometry fulfilling (\ref{eq:III_isometry}).
\end{cor}

\subsection{Saturation of the area law}

We finish this section by noticing that there exists another property that is implied by RFP tensors, and which is connected to the entanglement content of the states. We consider the von Neumann entropy of a block of $L$ spins contained in $|V^{(N)}(A)\rangle$. This entropy is defined as
 \be
 S_L^{(N)}(A) = - {\rm tr}\left[\rho_L \log_2(\rho_L)\right]:= S(\rho_L)
 \ee
where $\rho_L$ is the corresponding reduced state of $L$ spins. Since all MPS fulfill the area law, it is upper bounded by a constant. Furthermore, using the strong subadditivity inequality, translationally invariance, and the fact that $S_{(N)}^N=0$, it is simple to show that $S_L^{(N)}$ is an increasing function of $L$ for $N>2L$. Thus, in the limit $N\to \infty$ it will saturate to a constant. This motivates the following definition
 \begin{defn}
 We say that a tensor $A$ {\em saturates the area law} (SAL) if the generated MPS fulfill $S_1^{(N)}(A)=S_2^{(N)}(A)=\ldots S_{N/2}^{(N)}(A)$
 \end{defn}

We can establish a relationship between SAL and RFP as follows (the result is a trivial consequence of the general case of mixed states that we will analyze later on):

\begin{prop}
\label{ZCLandSALpure}
If a tensor $A$ in CF is RFP then it is SAL.
\end{prop}

Note that, as it will be clear from the next section when we study mixed states, the converse is not necessarily true. However, in that case we will impose ZCL as well in order to assess the equivalence, at least for some class of tensors. In the present pure case it does not make sense to add this condition, since it is already equivalent to RFP.

\section{Mixed States}\label{Sec:mixed-states}

As in the previous section we consider spins in a 1D chain, but now we focus on mixed states. That is, we consider a tensor, $M$, that generates translationally invariant density operators acting on the spins called matrix product density operators (MPDO) \cite{juanjoMPDO,guifreMPDO}. The tensor $M$ has four indices, two auxiliary (as a MPV) and two physical (as an operator acting on the Hilbert space of the spins). That is, for the MPDO corresponding to $N$ spins we write
 \be
 \label{eq:III_MPDOform}
 \rho^{(N)}(M) =  \raisebox{-8pt}{\includegraphics[height=3.2em]{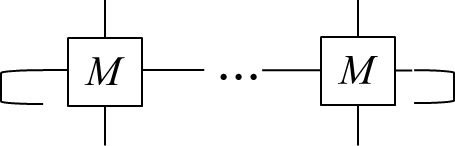}}
 \ee
The arrows indicate that the density operator is acting from bottom to top. Those operators are hermitian and positive semidefinite, ie. $\rho^{(N)}(M)=\rho^{(N)}(M)^\dagger \ge 0$.

In the following, we will define the renormalization fixed points (RFP) of the tensors $M$ generating families of MPDO. We will also analyze their connection with those that have "no length scale", by studing states with ZCL and also that saturate the area law (SAL) for the mutual information. As we will see, those conditions are not equivalent, but together coincide with the definition of RPF for some particular cases. We will characterize the tensors corresponding to ZCL and SAL and, finally, the RPF for the general case.

\subsection{Renormalization flow and renormalization fixed points}

As opposed to the previous section, it is in principle not clear what should be the right definition of RFP in this context, apart from the fact that it should capture the hand-waving intuition of "blocking until there is no length scale left in the system". The problem we face can be explained as follows. If we define a renormalization procedure as we did before, the dimension of the Hilbert space corresponding to one site can grow indefinitely. In the pure state case, due to the area law the relevant part of this space (i.e., the subspace where the state is supported) remains finite (its dimension is bounded by $D^2$), which allows one to describe the fixed points of the procedure within the MPV formalism. In the mixed case, this is not longer true (it is only the space of operators acting on that space which remains finite, but their action can involve the whole Hilbert space). Thus, in general, one should deal with local physical spaces whose Hilbert spaces have infinite dimensions, which takes us out of the MPV description.

A natural way around this problem is to define a renormalization flow that keeps the physical dimensions, very much in the spirit of standard renormalization group. This could be done in terms of a trace-preserving completely positive map (tpCPM), ${\cal S}$, that transforms the state of two spins, $\rho_{12}$ into a state of one, $\sigma_1$, but keeping the Hilbert space of the final spin equal to the original one: ${\cal S}(\rho_{12}) = \sigma_1$. This automatically induces a map in the tensors $M$ generating the MPDO. Defining
 \begin{subequations}
 \bea
 M_2(X)&=& \raisebox{-12pt}{\includegraphics[height=3.0em]{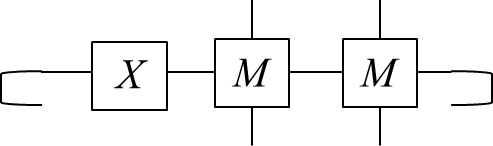}}\\
 N_1(X)&=&\raisebox{-12pt}{\includegraphics[height=3.0em]{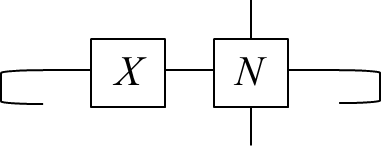}}\\
 \eea
 \end{subequations}
Then, we have ${\cal S}[M_2(X)]=N_1(X)$ for all $X$. The flow, as before, can be obtained by successive application of ${\cal S}$. The fixed point, a tensor $M$, should then satisfy
 \be
 \label{eq:Smap}
 {\cal S}[M_2(X)]=M_1(X).
 \ee
However, notice that the map may truncate part of the information, and thus make the procedure irreversible. In order to restore the reversibility, we will require the existence of another tpCPM, ${\cal T}$ mapping the state of one spin into that of two, i.e. ${\cal T}[M_1(X)]=N_2(X)$, such that for the fixed points it also fulfills
 \be
 \label{eq:Tmap}
 {\cal T}[M_1(X)]=M_2(X)
 \ee
for all $X$. This motivates the following

\begin{defn}
\label{RFPMixedTS}
Given the tensor $M$, in CF, generating a family of MPDOs, we say that it is a RFP if there exist two tpCPM, ${\cal T},{\cal S}$ acting on the physical indices, fulfilling (\ref{eq:Smap}) and (\ref{eq:Tmap}).
\end{defn}

Graphically,
 \be
 \label{TandSforsimple}
  {\includegraphics[height=4.8em]{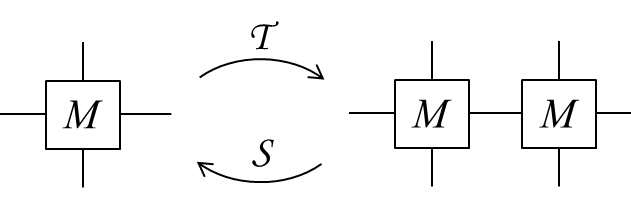}}
 \ee

The first observation is that this definition coincides with Definition \ref{defRFP} in the particular case in which the MPDO is pure, i.e. an MPS. In that case, the tensor $M$ is given by 

 \be
 \label{M-pure-case}
  {\includegraphics[height=8em]{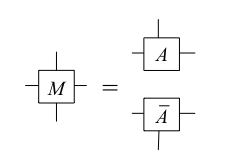}}
 \ee

Definition \ref{defRFP} is equivalent to the existence of a unitary $U$ acting on two sites and a pure state $e$ on one site so that 

\be
 \label{Def-equiv-pure}
  {\includegraphics[height=8em]{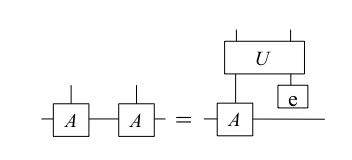}}
 \ee

Then we obtain (\ref{TandSforsimple}) for $M$ if we define ${\cal T}(X)=V(X\otimes |e\rangle\langle e|)V^\dagger$ and ${\cal S}(Y)=\tr_2(V^\dagger Y V)$ where $\tr_2$ traces the second system. Graphically,

\be
 \label{T-and-S-mixed-equiv-pure}
  {\includegraphics[height=11em]{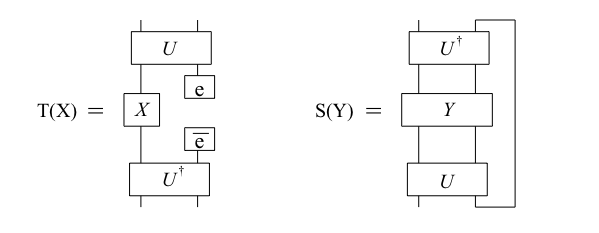}}
 \ee

Definition \ref{RFPMixedTS} appears also very naturally in the context of boundary theories.


\subsection{Boundary theories}\label{Section:Boundary}

Using the same tensor notation as introduced in Section \ref{section:graphical}, one can deal with systems arranged in 2D lattices (wlog the square lattice). If the tensors give rise to pure states, they are called Projected Entangled-Pair States (PEPS). That is, the tensor
\be \label{PEPS-single tensor}
 \raisebox{-12pt}{\includegraphics[height=2.8em]{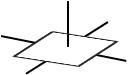}}
\ee
while contracted through a 2D square lattice of size $L_H\times L_V$ with periodic boundary conditions
\be \label{PEPS-square-lattice}
 \raisebox{-12pt}{\includegraphics[height=5.8em]{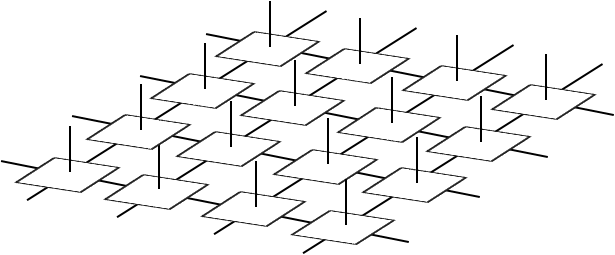}}
\ee
gives rise to a vector in $(\mathbb{C}^d)^{\otimes L_HL_V}$.

PEPS are known to be a good description of ground states of gapped local Hamiltonians and, since any ground state of a frutration-free commuting Hamiltonian is a PEPS, they are expected to represent exactly ground states of gapped RFP quantum phases in 2D.

In \cite{bulkboundary}, motivated by the seminal paper of Haldane and Li \cite{HaldaneLi}, the authors show an holographic correspondence for PEPS: an explicit isometry connecting the PEPS with a 1D mixed state living at the boundary. Interpreting such mixed state as $e^{-H_b}$, gives a 1D Hamiltonian   $H_b$ which is called the {\it  boundary theory} of the PEPS. Moreover, $e^{-H_b}$ is given explicitly in terms of the left and right fixed points of the transfer operator of the PEPS,
 \be \label{PEPS-transfer-operator} \raisebox{-12pt}{\includegraphics[height=5.8em]{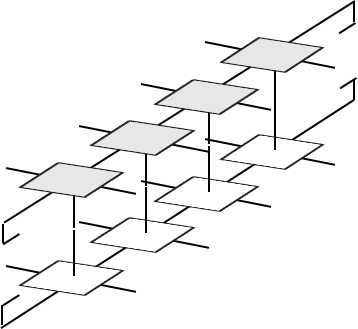}}
 \ee
It is clear from the way one computes expectation values in PEPS that such fixed points encode the correlations present in the PEPS. In an abuse of notation, here we will call any such fixed point a boundary theory.

As in any holographic correspondence, the aim of studying boundary theories is to provide with a dictionary translating relevant properties of the bulk to physically meaningful properties of  $H_b$. In \cite{bulkboundary, bulkboundconj} numerical evidence has been provided that (i) a gapped bulk corresponds to $H_b$ having quasi-local interactions; and (ii) the existence of topological order is reflected in some anomaly present in $H_b$, in the sense that $e^{-H_b}$ is supported only  in a subspace of the total Hilbert space that encodes the topological features of the bulk. The main results of this section can be seen as analytical proofs of such dictionary for the case of RFP.

From the point of view of boundary theories, there is a natural way to define RFP for mixed states, take a PEPS which is already the fixed point of some renormalization procedure and see how such property is transferred to the boundary. In this sense, let us consider a PEPS that is in an RFP, so that  its tensor verifies the following RFP property:
\be \label{PEPS-RFP}
 \raisebox{-12pt}{\includegraphics[height=5.8em]{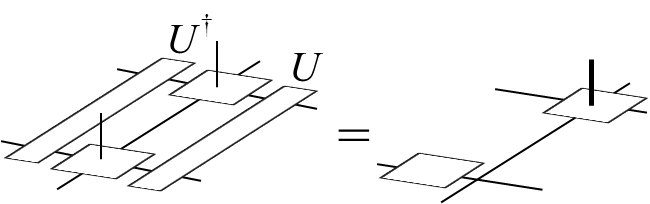}}
\ee
where equality means up to an isometry acting on the physical indices. 
In this case, the tensor associated to the boundary theory can be seen to be (using also that the PEPS is an RFP):
\be
 \raisebox{-12pt}{\includegraphics[height=5.0em]{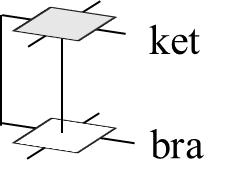}}
\ee
and property (\ref{PEPS-RFP}) maps to
\be
 \raisebox{-50pt}{\includegraphics[height=19.0em]{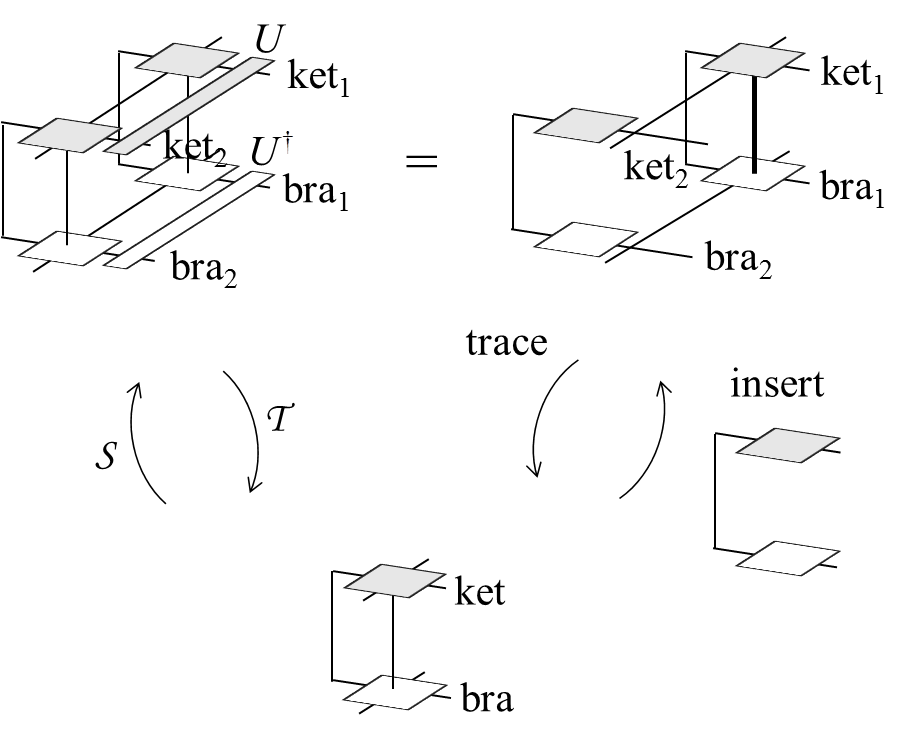}}
\ee
which gives the desired $T$ and $S$ of Definition \ref{RFPMixedTS}.
\subsection{Zero correlation length}

We have seen that in the MPS case ZCL is equivalent to RFP. Since the pure case is a particular case of the mixed one, one may try to see whether the corresponding generalization of ZCL is enough to capture the RFP property.

\begin{defn}
\label{DefinitionZCL}
A tensor $M$ generating MPDO is said to have zero correlation length (ZCL) if
 $$\raisebox{-12pt}{\includegraphics[height=4em]{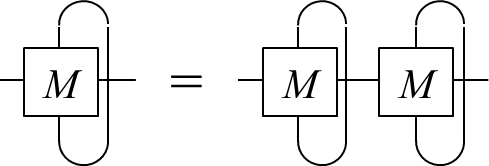}}$$
\end{defn}

This is the natural extension of the pure state case of Theorem \ref{TheoremZCLPure}.
Indeed, it is clear that if we compute correlation functions using a  MPDO generated by a tensor $M$ with ZCL, they will be length-independent.

In order to characterize the tensors fulfilling this property, we will write the tensor $M$ generating MPDOs in terms of another tensor, $A$, generating a purification.\footnote{This may not be always possible, see \cite{Gemma}.} Graphically, we have
  \be
  \label{Psipuri}
  {\includegraphics[height=7em]{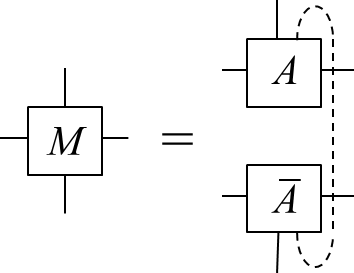}}
 \ee
In terms of the MPDO, the interpretation is that apart from the $N$ spins, we have $N$ ancillary systems, each of them associated to one spin, such that
 \be
 \label{eq:MPDO-Puri-1}
 \rho^{(N)}_p(M)={\rm tr}_a(|\Psi^{(N)}(A)\rangle\langle\Psi^{(N)}(A)|),
 \ee
where the trace is with respect to the ancillary systems, and $\Psi$ is a MPS describing the spins and the ancillas.

Now we can apply the renormalization procedure for pure states to the tensor $A$, which induces a notion of RFP for the tensors corresponding to mixed states:

\begin{defn}\label{def:Puri-RFP}
We call a tensor $M$ a Purification Renormalization Fixed Point (PRFP) if it can be written as (\ref{eq:MPDO-Puri-1}), where $A$ is a RFP.
\end{defn}

Note that the tensor $A$ in (\ref{Psipuri}) will have the form (\ref{III_CFI_RFP}), but now with two physical indices corresponding to the spin and the ancilla. When tracing over the ancilla indices,
this will give rise to a superoperator, ${\cal E}$, which is a tpCPM.
In particular, we can write
 \be
\label{PuriRFP}
 \rho^{(N)} = \sum_{j,j'=1}^g  \sum_{q=1}^{r_j}  \sum_{q'=1}^{r_{j'}} e^{iN(\phi_{j,q}-\phi_{j',q'})} \rho^{(N)}_p(A_j,A_{j'}),
 \ee
where
 \be\label{PuriRFP-BNT}
 \rho^{(N)}_p(A_j,A_{j'})= {\cal E}^{\otimes N} \left[|\varphi_j\rangle\langle \varphi_{j'}|^{\otimes N}\right],
 \ee
and $\varphi$ is given in (\ref{eq:III_varphi}). Note that {\it any} MPDO with a MPS purification can be written in the same way if $\mathcal{E}$ is asked to be only a CPM (not necessarily trace-preserving).

Thus, we can give the following characterization:

\begin{thm}
The following statements are equivalent for a tensor $M$ fulfilling (\ref{eq:MPDO-Puri-1}):
 \begin{description}
\item[(i)] $M$ is a PRFP
\item[(ii)] $M$ has ZCL
\item[(iii)] The density operators generated by $M$ have the form (\ref{PuriRFP})
\end{description}
\end{thm}

Even though it has a nice characterization in terms of zero-correlation length, Definition \ref{def:Puri-RFP} has two main problems. The first one is that it only applies to MPDOs with a purification and, as shown in \cite{Gemma}, this does not need to be the case. The second, and most important, is that it is too weak. As shown in Example \ref{ExZCLnoSAL} below, for mixed states ZCL does not immediately imply that there is no length scale in the system. In particular, it does not imply the saturation of the area law for the mutual information. This is different than for pure states, where the saturation of the area law is implied by ZCL (Proposition \ref{ZCLandSALpure}).

\subsection{Mutual information. Saturation of the area law}
\label{MixedMutual}

Throughout this Section, we will need to consider normalized states in order to talk meaningfully about entropic quantities. As a convention, whenever we consider one such entropic quantity for an MPDO $\rho^{(N)}(M)$, we will be referring to its normalized version 
$$\frac{\rho^{(N)}(M)}{{\rm tr}[\rho^{(N)}(M)]}.$$

If we consider a chain with $N$ spins, the mutual information between a part of the chain of length $L$ and the rest is defined as
 \be
 I_L = S_L + S_{N-L}- S_N
 \ee
where $S_L$ is the von Neumann entropy of the reduced state of $L$ neighboring spins. For pure states, it coincides with (twice) the entanglement entropy of that part of the chain with the rest. It is not difficult to check (see \ref{AppendixMixed}) that it fulfills the following monotonicity property

\begin{prop}
\label{PropILILp1}
For any MPDO and $L<\lfloor N/2\rfloor$, $I_{L}\le I_{L+1}$ and
 \be
 \lim_{L\to\infty} \lim_{N\to\infty} I_L =I_\infty <\infty.
 \ee
\end{prop}

This entails that the mutual information must increase and saturate with $L$, thus providing a length scale to the system. In a RFP one expects then a saturation of the area law as we define now

\begin{defn}\label{def:area-law}
A tensor $M$ generating MPDO fulfilling $I_1=I_2=...$ is said to verify {\em saturation of the area law} (SAL).
\end{defn}

Below we will show examples of tensors which have ZCL but no SAL, nor the converse. In view of that, it makes sense to consider tensors with ZCL and SAL. According to their definition (\ref{def:area-law}), they will have to fulfill $I_L=I_{L+1}$, or, equivalently, $S_L+S_{N-L}=S_{L+1}-S_{N-L-1}$ for $L<\lfloor N/2\rfloor$. Note that this condition involves reduced density operators in which we have traced several spins. Thus, if the MPDO (\ref{eq:III_MPDOform}) contains terms that vanish whenever one traces few spins, then they will not be visible in the conditions we derive. More specifically, let us consider the tensor $M$ generating the MPDO as a MPV where the two indices (bra and ket) corresponding to the action on a spin are taken as physical indices. By assuming that we have blocked enough number of spins, we can write it in CF in terms of some BNT, $M_k$. Let us define
  \be
  {\includegraphics[height=4em]{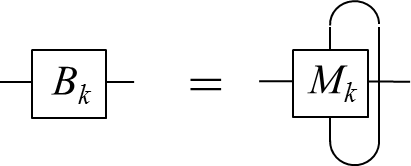}}
 \ee
There may be some nilpotent $B_k$, meaning that ${\rm tr}_R\rho^{(N)}(M_k)=0$, where we have traced $R\le D$ sites, and $N$ is sufficiently large. Thus, whenever we consider the entropy of a reduced state of the MPDO in which we have traced more than $D$ spins, we cannot state anything about the nilpotent terms in the BNT. This motivates the following:

\begin{defn}
A tensor generating MPDO's is simple if none of the elements in a BNT is nilpotent.
\end{defn}

We aim now to provide the equivalent of Theorem \ref{thm:main-MPS} for simple mixed states. In that theorem, three properties were shown to be equivalent for an MPS in CF: \emph{(i)} being a RFP of some renormalization flow, characterized by the fact that there is a unitary (hence reversible) way to create one tensor out of two; \emph{(ii)} the absence of some concrete type of length scales in the system, CID and LO in that case; \emph{(iii)} its parent Hamiltonian being commuting.

The role of ZCL (equivalently CID and LO) will be played now by ZCL and SAL. While we have an analogue of \emph{(i)} through the definition of RFP \ref{RFPMixedTS}, we need to find an analogue of \emph{(iii)}.  A natural choice, which also fits very well if one understands the given MPDO as the boundary theory of a 2D system (see Section \ref{Section:Boundary}), is the fact that the MPDO is of a special form.

\begin{defn}
\label{defrhoNComm}
We say that a density operator of $N$ spins is a Gibbs state of a nearest-neighbor commuting Hamiltionian (GSNNCH), if it can be written as
 \be
 \label{rhoNComm}
 \rho^{(N)}\propto \oplus_{x} n_x e^{-\sum_{j=1}^N \tau_{j}(h^{(x)})}.
 \ee
where $n_x$ are natural numbers and the $h^{(x)}$ acts on the first two spins, $\tau_j$ translates the spins by an amount $j$, and $[h^{(x)},\tau_1(h^{(x)})]=0$.
\end{defn}

The direct sum indicates that the global Hilbert space ${\cal H}= \oplus_x {\cal H}_x$, where
${\cal H}_x=\otimes_n {\cal H}_n^{x}$, and for each $n$, the ${\cal H}_n^{(x)}$ are orthogonal. The density operator (\ref{rhoNComm}) is translationally invariant, where the spin $N+1$ is identified with the first.

We will assume throughout the text that equation \eqref{rhoNComm} includes the limiting case in which $h=h_1+\beta h_2$ with $h_2\ge 0$, $[h_1,h_2]=0$ and $\beta\rightarrow \infty$, so that $e^{-h}=P_2e^{-h_1}$, with $P_2$ the projector onto the ground space of $h_2$. This shows that $e^{-h}$ can be considered  as an arbitrary positive semidefinite operator on the first two spins. It is clear then that, with this generalization, form \eqref{rhoNComm} is equivalent to 
 \be
 \label{rhoNCommv2}
 \rho^{(N)}\propto \oplus_{x} n_x \prod_{j=1}^N \tau_{j}(B^{(x)}).
 \ee
where $n_x$ are natural numbers, the $B^{(x)}\ge 0$ acts on the first two spins (and are identity on the rest), $\tau_j$ translates the spins by an amount $j$, and $[B^{(x)},\tau_1(B^{(x)})]=0$.

In the following we will establish the main result of this subsection. We will consider a simple tensor, $K$, in biCF, with BNT ${\cal K}_j$ and corresponding coefficients $\mu_{j,q}$ [cf (\ref{eq:II_ABasicTensors})]. In order to emphasize that the tensor is simple, we will use $K$ instead of $M$, and $\sigma^{(N)}(K)$ instead of $\rho^{(N)}(K)$.

\begin{thm}\label{thm:main-simple}
Let us call $M$ the tensor obtained by blocking two sites:
\be
 \label{K=MM}
  \raisebox{-12pt}{\includegraphics[height=3.5em]{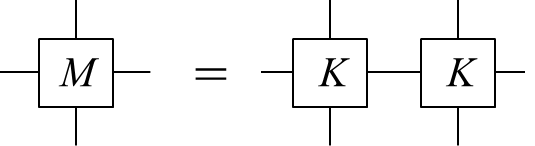}}
\ee

Consider the following statements:

\begin{description}
\item[(i)] $K$ is a RFP.
\item[(ii)] $K$ has ZCL and SAL.
\item[(iii)]  $\sigma^{(N)}(K)$ is a GSNNCH and has ZCL
\item[(iv)] The elements of a BNT are supported on different subspaces, i.e.
 \be
 \label{KxKy=0}
  \raisebox{-24pt}{\includegraphics[height=5.5em]{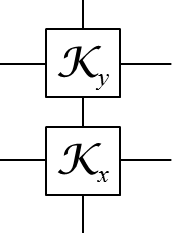}}=0 \quad {\text if} \; x\ne y,
\ee
and can always be chosen such that $\sigma^{(N)}({\cal K}_j)\ge 0$ for all $N$. Furthermore, for each BNT element, ${\cal K}$,
there exists an isometry $U$ such that
\be
 \label{UkU=rl}
  \raisebox{-12pt}{\includegraphics[height=8em]{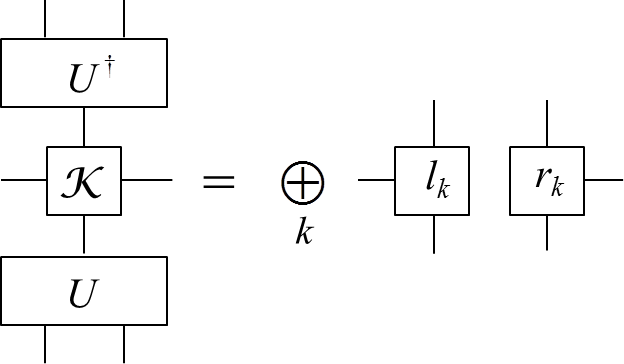}}
\ee
with
\be
 \label{etakhetc}
  \eta_{k,h}=\raisebox{-12pt}{\includegraphics[height=3em]{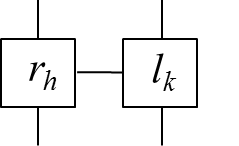}}\ge 0
\ee
and
 \begin{equation}
 \label{tralktrrk}
T_{k,h}:= {\rm tr}(\eta_{k,h})=a_kb_h\\
 \end{equation}
with
 \begin{equation}
 \label{PsiPhi}
 \sum_k a_k b_k=1.
 \end{equation}
 \item[(v)] $M$ is a RFP.
\end{description}
Then $i \Rightarrow ii \Leftrightarrow iii \Rightarrow iv \Rightarrow v$.  
\end{thm}

The proof will be given in \ref{AppendixMixed}. We finish this section giving two examples:

\begin{ex}
\label{ExZCLnoSAL}
(state with ZCL but no SAL):
Let us consider the state of four spins, each of them composed of two qubits. We identify by $(n,l)$ and $(n,r)$ the two qubits corresponding to the $n$-th spin (one to the left and the other to the right). Each qubit is maximally entangled with its neighbor. That is, $nr$ and $(n+1)l$ are in the state $(|0,0\rangle+|1,1\rangle)/\sqrt{2}$, where the spin $5$ is identified with the first. We apply to each spin the tpCPM
 \be
 {\cal E}_n(X)= p \sigma^{(n,l)}_x\otimes\sigma^{(n,l)}_x X \sigma^{(n,l)}_x\otimes\sigma^{(n,l)}_x + (1-p) X.
 \ee
It is simple to show that $I_2>I_1$ for any $p\ne 0,1/2$. For instance, taking $p=0.25$, one can see that $S_1=2$, $S_2=2.9544$, $S_3=3.8802$, and $S_4= 2.7839$, where $S_i$ is the entropy of $i$ neighboring spins. Thus, $I_1=3.0963$ and $I_2= 3.1250$ and it does not saturate the area law. Furthermore, it has obviously ZCL since it fulfills (\ref{PuriRFP}).
\end{ex}

\begin{ex}
\label{ExSALnoZCL}
(state with SAL but no ZCL):
We just have to take the same configuration as in the previous example, with $N$ spins, each of them composed of two qubits and define
 \be
  \sigma\propto\sum_{k_1,\ldots,k_N=0}^1 \otimes_{n=1}^N \eta_{k_n,k_{n+1}}
  \ee
where $\eta_{k_n,k_{n+1}}$ acts on the qubits $(n,r)$ and $((n+1),l)$ (identifying $N+1$ with 1). We take
 \begin{subequations}
 \begin{eqnarray}
 \eta_{0,0}&=&|0\rangle\langle 0| \otimes |0\rangle\langle 0| ,\\
 \eta_{0,1}&=&\frac{1}{2}|0\rangle\langle 0| \otimes |1\rangle\langle 1| ,\\
 \eta_{1,0}&=&\frac{1}{2}|1\rangle\langle 1| \otimes |0\rangle\langle 0| ,\\
 \eta_{1,1}&=&|1\rangle\langle 1| \otimes |1\rangle\langle 1|
 \end{eqnarray}
It is easy to show that this state fulfills the criterion of SAL. However, it does not have ZCL. A similar example for pure states also shows that in the case of MPS, saturation of the entanglement entropy does not suffice to characterize RFP.
 \end{subequations}
\end{ex}

\subsection{General case}

We have seen in the previous section that the RFP tensors generating MPDO can be fully characterized in terms SAL and ZCL for simple tensors. Those are the ones in which, when we trace one or more spins, we do not lose information about the tensor. Since, as shown in Theorem \ref{thm:main-simple}, the MPDO they generate are basically Gibbs states of commuting nearest neighbor Hamiltonians, they are, in some sense, trivial. For instance, they do not include the state considered in the following

\begin{ex}
Let us consider a tensor, $M$, with $d=D=2$ and all components equal to zero except for: $1=M_{00}^{00}=M_{00}^{11}=M_{11}^{00}=-M_{11}^{11}$. The corresponding MPDO is
 \be
 \rho^{(N)}(M)= \Id^{\otimes N} + \sigma_z^{\otimes N},
 \ee
where $\sigma_z=|0\rangle\langle 0| - |1\rangle\langle 1|$. This is precisely the boundary state corresponding to the toric code \cite{Toric,bulkboundconj} (Section \ref{Section:Boundary}). If we trace one spin, the reduced state becomes proportional to the identity, i.e. one loses all the information about the $\sigma_z$ part. Note that this tensor is SAL and has ZCL. However, it cannot be expressed in the form given by (\ref{rhoNComm}).
They nevertheless are RFP since one can easily find ${\cal S}$ and ${\cal T}$ [see (\ref{RFPMixedTS})].
\end{ex}

In this section we study RFP tensors that are general, and thus not necessarily simple. As we will see, their characterization will give rise to non-trivial solutions, which connect with unitary fusion categories and topologically ordered models in 2D.

In order to characterize the RFP property as defined in  (\ref{TandSforsimple}) for  tensors $M$ generating MPDO,  we have to use the CF of $M$ but in the vertical direction. That is, we will consider the tensor as generating MPV by concatenating it vertically, and thus by exchanging the notion of physical and virtual indices. We have

\begin{prop}\label{Prop:IV.12}
A tensor $\tilde M$ in CF, generating MPDO, is also in CF vertically. Moreover, there exists an isometry $U$ such that
 \be
 \label{UMU}
 U \tilde M U^\dagger = \bigoplus_\alpha \mu_\alpha \otimes M_\alpha,
 \ee
where $\mu_\alpha$ are diagonal and positive matrices, and $\{M_\alpha\}_\alpha$ is a BNT.
\end{prop}

In graphical form, (\ref{UMU}) simply means
 \be
 \label{graphUMU}
 \raisebox{-12pt}{\includegraphics[height=8em]{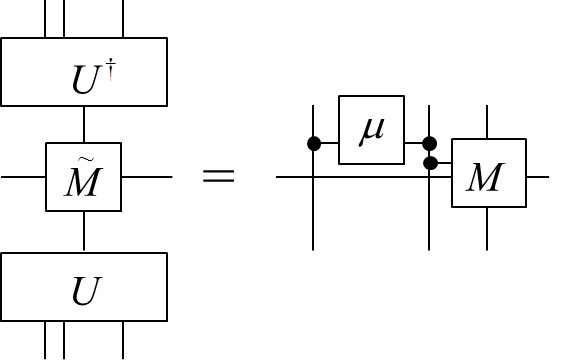}}
 \ee
In what follows, we will ignore the isometry, as it just changes the physical basis on each site.

Given a tensor $M$ we define the operator $O_L(M)={\rm tr}(MM...M)$, where $M$ is multiplied $L$ times.
Graphically, we have
 \be
 O_L(M)= \raisebox{-50pt}{\includegraphics[height=10em]{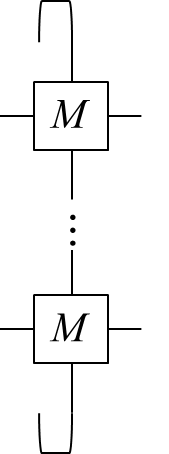}}
 \ee
We also consider the corresponding operators for the BNT of M, $O_L(M_\alpha)$.

Now, we can give the characterization of  RFP, which is the main result of this paper.

\begin{thm}\label{thm:IV.13}
Given a tensor $M$ in CF that generates MPDO, the following statements are equivalent:
 \begin{itemize}
 \item[(i)] $M$ is a RFP.
 \item[(ii)] There exists a set of diagonal matrices, $\chi_{\alpha,\beta,\gamma}$, with positive elements, such that for each $L$ the operators $O_L(M_\alpha)$ linearly span an algebra with structure coefficients $c_{\alpha,\beta,\gamma}^{(L)}={\rm tr}(\chi^L_{\alpha,\beta,\gamma})$:
 \be
 \label{eq:algebra}O_L(M_\alpha)O_L(M_\beta)=\sum_\gamma c_{\alpha,\beta,\gamma}^{(L)}O_L(M_\gamma),
 \ee
 and
  \be
 \label{idempotent}
 m_\gamma = \sum_{\alpha,\beta} c^{(1)}_{\alpha,\beta,\gamma} m_\alpha m_\beta
 \ee
 where $m_\alpha={\rm tr}(\mu_\alpha)$. That is, the vector $(m_\alpha)_\alpha$ is an idempotent for the ``multiplication'' induced by $c^{(1)}$.
 \item[(iii)] There exist isometries, $U_{\alpha,\beta}$, such that
 \be
 \label{Ualphabeta}
 U_{\alpha,\beta} M_\alpha M_\beta U_{\alpha,\beta}^\dagger = \oplus_\gamma \chi_{\alpha,\beta,\gamma} \otimes M_\gamma,
 \ee
where $\chi_{\alpha,\beta,\gamma}$ are the same as in the previous statement, and (\ref{idempotent}) is fulfilled.
 \end{itemize}
\end{thm}

This result is very appealing from the perspective of 2D topologically ordered models and boundary theories (Section \ref{Section:Boundary}), since it has been recently proven in \cite{FrankStrings} that starting from (\ref{eq:algebra}) and (\ref{idempotent}), for the particular case of $c_{\alpha,\beta,\gamma}^{(L)}\in \mathbb{N}$ and independent of $L$,  one recovers in a direct and elegant way all known non-chiral topologically ordered phases together with a description of their excitations. The RFP MPDO we started with correspond exactly with the boundary theories of the RFP representatives of such topological phases \cite{FrankStrings}. This connection also gives immediately a large family of non-trivial RFP MPDO. We leave as an interesting open question whether there are more examples than those. More concretely, whether there exist RFP $M$ for which $c_{\alpha,\beta,\gamma}^{(L)}$ depends on $L$.

As showed in \cite{FrankStrings}, the associativity of the multiplication in (\ref{eq:algebra}) imposes restrictions on the tensor $X=(\chi_{k,\alpha,\beta,\gamma})_{\alpha,\beta,\gamma,k}$ that, by (\ref{Ualphabeta}), give a pentagon-like equation for the $U_{\alpha,\beta}$. The solutions of such equation using F-symbols of unitary fusion categories give the family of RFP MPDO commented above: boundary theories of Levin-Wen string net models. In this sense, the {\it topological content} of a RFP MPDO lies in the unitaries $U_{\alpha,\beta}$ of (\ref{Ualphabeta}).

Now, by applying equation (\ref{Ualphabeta}) recursively, one may write any RFP MPDO given by a tensor $M$ as locally unitary equivalent --via the unitaries $U$ in (\ref{UMU})-- to $(\mu^{\otimes N})\tilde{U}Q\tilde{U}^\dagger$, where $\mu=\sum_\alpha \mu_\alpha \otimes \ket{\alpha}\bra{\alpha}$ is the tensor defined in (\ref{graphUMU}); both $\mu^{\otimes N}$ and $\tilde{U}Q\tilde{U}^\dagger$ are completed by $\Id$ in the spaces in which they do not act;
\be
[\mu^{\otimes N}, \tilde{U}Q\tilde{U}^\dagger]=0\; ;
\ee

\be
 \label{Projector-Q}
 Q=\quad\quad  \raisebox{-30pt}{\includegraphics[height=6em]{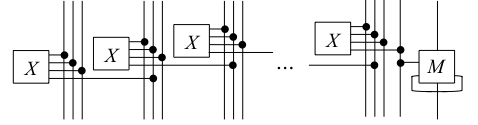}}\quad ,
\ee
with each $X=(\chi_{\alpha,\beta,\gamma, k})_{k,\alpha,\beta,\gamma}$ having assigned their indices in the order $ k,\alpha,\beta,\gamma$; and $\tilde{U}$ is a sequential circuit of unitaries made out of the $U_{\alpha,\beta}$'s in (\ref{Ualphabeta}).

If  the structure coefficients $c_{\alpha,\beta,\gamma}^{(L)}={\rm tr}(\chi^L_{\alpha,\beta,\gamma})$ are independent of $L$ one can easily show by Lemma \ref{Lem:app_simple} that the $\chi_{\alpha,\beta,\gamma,k}\in \{0,1\}$  and therefore $c_{\alpha,\beta,\gamma}^{(L)}\in \mathbb{N}$. In this case, if the final matrices $\tr(M_{\gamma})$ at the RHS of (\ref{Projector-Q}) are  projectors, we have that $Q$ in (\ref{Projector-Q}) is also a projector. The same applies then to $\tilde{U}Q\tilde{U}^\dagger$. By making an spectral decomposition of the $\tr(M_{\gamma})$'s  we can reduce to such case and obtain the following theorem.

\begin{thm}
Start with a  RFP tensor $M$  defining  MPDO so that $c_{\alpha,\beta,\gamma}^{(L)}$  is independent of $L$. Then the MPDO
$$\rho^{(N)}(M)=\sum_{i=1}^d \lambda_i P^{(N)}_i e^{-H_N}$$ where $d$ is the single site Hilbert dimension, $P_i{(N)}$ are projectors which encode the {\em topological properties of the tensor, if interpreted as a boundary theory}, $H_N=\sum_{i=1}^Nh_{i,i+1}$ is translationally invariant, nearest-neighbor and commuting $[h_{i-1,i},h_{i,i+1}]=0$, and $[P_i,e^{-H}]=0$ for all $i$.
\end{thm}

This result can be seen then as an analytical proof, for the case of RFP, of the conjecture stated in \cite{bulkboundconj} about how boundary theories look like for topologically ordered models (Section \ref{Section:Boundary}).

\section{Conclusion}

In this paper we have given a full characterization of MPS and MPDO that are RFP in terms of their defining tensor. We have also shown that being a RFP can be also characterized by physical notions: zero-correlation length, saturation of the area law and being the ground or thermal state of commuting Hamiltionians.

The motivation of the paper, and arguably its main interest, lies on its connection with boundary theories and the classification program of 2D topological models on the lattice. It is believed that within each gapped topological quantum phase, there exists one representative that is the fixed point of some renormalization transformation. At a gapped RFP one expects an exact description of its ground state in terms of PEPS and an associated boundary theory which is a RFP MPDO. Here we have illustrated this for a concrete and natural definition of RFP in 2D but we expect the same result for any meaningful definition of RFP in 2D (we will analyze this in full detail in a forthcoming paper).  Since all properties of the PEPS are encoded on its boundary via the boundary-bulk correspondence of \cite{bulkboundary}, one can understand this paper as a way to classify 2D quantum phases on their boundary theories. From this perspective, our main result (Theorem \ref{thm:IV.13}) is very appealing, since it shows how the topological content of the bulk, at a RFP, can be clearly distilled from the canonical form of the MPDO that appears at the boundary.  Moreover, when complemented with \cite{FrankStrings}, our main result suggests that all quantum phases in 2D on the lattice seem to be already covered by Levin-Wen string net models and their generalizations.

\

JIC acknowledges support from the EU Integrated projects SIQS, and the DFG through the excellence cluster NIM.
DPG acknowledges support from
Comunidad de Madrid (grant QUITEMAD+-CM, ref. S2013/ICE-2801) and MINECO (grant no. MTM2014-54240-P).
This project has received funding from the European Research Council (ERC) under the European Union's Horizon 2020 research and innovation program (grant agreement GAPS, No 648913).
This work was made possible through the support of grant \#48322 from the John Templeton Foundation. The opinions expressed in this publication are those of the authors and do not necessarily reflect the views of the John Templeton Foundation.
NS acknowledges funding by the European Research Council (ERC) under the European Union's Horizon 2020 research and innovation program (grant agreement WASCOSYS, No 636201).
FV acknowledges funding by EU grant SIQS and ERC
grant QUERG, the Odysseus grant from the Research Foundation Flanders (FWO)
and the Austrian FWF SFB grants FoQuS and ViCoM.
This research was supported in part by the National Science Foundation under Grant No. NSF PHY11-25915.

We thank Michael Mari\"en for very useful discussions. We thank Yoshiko Ogata for her help to detect and correct an error in the statement and proof of Theorem \ref{thm:main-simple}. We also thank KITP (UC Santa Barbara), Perimeter Institute for Theoretical Physics and IQIM (Caltech) for their kind hospitality. This work was largely developed during the visits of the authors to those institutions.

\appendix


\renewcommand{\thethm}{\Alph{section}.\arabic{thm}}


\section{Proofs of Section \ref{Sec:MPV}}\label{AppendixMPV}

To prove the Fundamental Theorem of MPV and the associated results stated in Section \ref{Sec:MPV} we will need some lemmas that we will state and prove below. Before entering into that, and to connect also with previous (more stringent) definitions of canonical forms for MPS, we will introduce the following definition:

\begin{defn}
We say that a set of matrices $A^i$ are in a {\em canonical form II} (CFII), if they are in CF and the corresponding CPMs (\ref{Ek}) are trace-preserving, and have a full-rank diagonal fixed point.
\end{defn}

The two last properties can be expressed as
 \begin{subequations}
 \label{eq:II_TPLambda}
 \bea
 \label{TP}
 \sum_{i=1}^d  A^{i\dagger}_k  A^i_k &=& \Id_{D_k},\\
 \label{Lambda}
 {\cal E}_k(\Lambda_k) &=& \sum_{i=1}^d A^{i}_k \Lambda_k A^{i\dagger}_k = \Lambda_k>0,
 \eea
 \end{subequations}
where $\Lambda_k$ is a strictly positive definite diagonal $D_k\times D_k$  matrix. In \cite{Fan92, David2006} it was shown that for any $A$ in CF, it is always possible to find a non-singular gauge matrix $X$, such that
 \be
 \label{eq:II_XAX}
 {\includegraphics[height=3.5em]{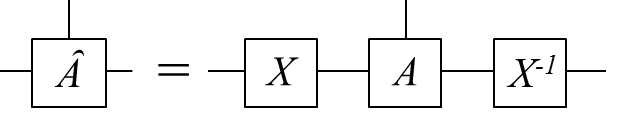}}
 \ee
is also in CF and fulfills (\ref{eq:II_TPLambda}). For that one just has to determine the fixed point (i.e., the eigenvector corresponding to the largest eigenvalues in magnitude) of each CPM ${\cal E}_k$.

\begin{lem}
\label{equalMPS}
Given two NMPV, $V_{a,b}$, generated by two NT with corresponding $D_\alpha\times D_\alpha$ matrices $A_{\alpha}^i$ ($\alpha=a,b$), then
 \begin{subequations}
 \bea
 \lim_{N\to\infty} \langle V_\alpha|V_\alpha\rangle&=&1,\\
 \label{scprod}
 \lim_{N\to\infty} |\langle V_b|V_a\rangle|&=& \text{ 0 or 1}.
 \eea
 \end{subequations}
In the latter case, $D_a=D_b$ and there exists a non-singular matrix, $X$, and a phase $\phi$ so that $A^i_b= e^{i\phi} X  A^i_a  X^{-1}$.
\end{lem}

\begin{proof}
Since  $V_{a,b}$ are NMPV, the comment above shows that we can always find non-singular matrices, $Y_{a,b}$, such that
 \be
 \label{rel}
 \tilde A^i_\alpha=Y_\alpha A^i Y_\alpha^{-1}
 \ee
are in a CFII and generate the same MPV. We will use those matrices to prove the proposition. In particular, we will show that if $|\langle V_b|V_a\rangle|\to 1$ then there exists a unitary $X$ such that $X A^i_b X^{-1}= e^{i\phi} A^i_a$. The result for the original NT follows automatically from (\ref{rel}). First, $\langle V_\alpha|V_\alpha\rangle={\rm tr}(\mathbb{E}_\alpha^N)$ where $\mathbb{E}_\alpha=\sum_{i=1}^d A_\alpha^i\otimes \bar A_\alpha^i$. This is nothing but the matrix representation of the CPM built out of the matrices $A^i_\alpha$, and thus has the same spectrum. Since $A_\alpha$ is a NT, $\mathbb{E}_\alpha$ has a unique eigenvalue of magnitude 1, and the rest are strictly smaller, from which follows that the norm tends to one for $N\to\infty$. Second, $\langle V_b|V_a\rangle={\rm tr}(\mathbb{E}_{1,2}^N)$, where
$\mathbb{E}_{ab}=\sum_{i=1}^d A_a^i\otimes \bar A_b^i$. Let us consider the eigenvalue equation for this matrix written as a linear map,
 \be
 \sum_{i=1}^d A_b^{i\dagger} X A_a^{i}=\lambda X.
 \ee
Using the fact that the fixed point of ${\cal E}_a$, $\Lambda_a$, is strictly positive, we have
 \bea
 &&|\lambda {\rm tr}(X \Lambda_a X^\dagger)|^2 = \left|\sum_{i=1}^d {\rm tr}(X A^i_a \sqrt{\Lambda_a} \sqrt{\Lambda_a} X^\dagger A_b^{i\dagger})\right|^2 \nonumber\\
 &\le& \sum_{i=1}^d {\rm tr}\left(A^i_b X \Lambda_a X^\dagger A_b^{i\dagger}\right)
 \sum_{j=1}^d {\rm tr}\left(X A^j_a \Lambda_a A_a^{j\dagger}X^\dagger\right)\nonumber\\
 && \le |{\rm tr}(X\Lambda_a X^\dagger)|^2,\nonumber
 \eea
where we have used (\ref{TP},\ref{Lambda}) and the Cauchy-Schwarz inequality. Since $\Lambda_a$ is full rank, $|\lambda|< 1$ or $|\lambda|=1$, giving rise to $0$ and $1$ in (\ref{scprod}), respectively. Let us explore the latter case $|\lambda|=1$. This requires that the Cauchy-Schwarz if fulfilled with an equal sign, and since $\Lambda_a$ is non-singular,
 \be
 A^i_b X = \alpha X A^i_a.
 \ee
Using $X=\sum_{i=1}^d A^{i\dagger}_b A^i_b X =\alpha \sum_{i=1}^d A^{i\dagger}_aXA^i_b=\lambda \alpha X$, and thus $|\alpha|=1$. Then, $\sum_{i=1}^d A^{i\dagger}_aX^\dagger X A^i_a=X^\dagger X =\Id$, and thus $X$ is an isometry. Furthermore, $X^\dagger A_b^i X = \alpha X^\dagger X A^i_a=\alpha A^i_a$. Let us assume that $D_a> D_b$; then there is a subspace
of dimension $D_a-D_b$ such that $A^i_a$ vanishes on it, which is impossible given that $A_a$ is a NT. Thus, $D_a=D_b$, and $X$ is unitary.
\end{proof}

From there we obtain the following two trivial corollaries:

\begin{cor}
\label{eqV}
Given two NT generating the NMPV, $V_{a,b}$, either $\langle V_a|V_b\rangle \to 0$ in the limit $N\to \infty$, or
 \be
 |V_b\rangle = e^{i\phi N} |V_a\rangle
 \ee
for all $N$.
\end{cor}

\begin{cor}
\label{Lem1}
Any set of NMPV, $\{V_j\}_{j=1}^g$, fulfilling $\langle V_j|V_{j'}\rangle\to \delta_{j,j'}$ for $N\to\infty$, is linearly independent for $N$ sufficiently large.
\end{cor}

We can now prove Proposition \ref{prop:char-BNT}, that we restate here for the convenience of the reader

\begin{prop*}
The tensors $A_j$ ($j=1,\ldots,g$) form a BNT of $A$ iff: (i) for all NT appearing in its CF (\ref{eq:II_CF1}), $\tilde A_k$, there exists a $j$, a non-singular matrix $X_k$, and some phase $\phi_k$ such that \be
 \label{eq-Appendix:II:A=XAX}
 \tilde A^i_k=e^{i\phi_k} X_k A^i_j X_k^{-1};
 \ee (ii) the set is minimal, in the sense that for any element $A_j$, there is no other $j'$ for which (\ref{eq-Appendix:II:A=XAX}) is possible.
\end{prop*}

\begin{proof}
The result is an immediate consequence of the following procedure to construct a BNT for any given tensor $A$. First, express $A$ in CF. Let us denote by $\tilde A_k$ the NT appearing in the decomposition (\ref{eq:II_CF1}). Then, take $A_1:=\tilde A_1$. Second, if $\langle V^{(N)}(\tilde A_2)|V^{(N)}(A_1)\rangle\to 0$ as $N\to \infty$, take $A_2:=\tilde  A_2$ and otherwise do not incorporate $\tilde  A_2$ to the set, since according to Corollary \ref{eqV} $|V^{(N)}(\tilde A_2)\rangle = e^{i \phi_2 N} |\tilde V^{(N)}(A_1)\rangle$. Then, continue in the same way; that is, if $\langle V^{(N)}(\tilde A_k)|\tilde V^{(N)}(A_j)\rangle\to 0$ for all $j=1,\ldots,J$, then $A_{J+1}:=\tilde A_k$. After this procedure, the BNT will be composed of a subset of $\tilde A_k$, say $A_j=\tilde{A}_{k_j}$, with $j=1,\ldots,g$. Note that if a NT, $\tilde A_k$, is not included in the BNT it must fulfill, according to Proposition \ref{equalMPS}, $\langle V^{(N)}(\tilde A_k)|\tilde V^{(N)}(A_j)\rangle\to 1$ for some $j$, and thus there exist a non-singular matrix $X_k$, and some phase $\phi_k$ such that (\ref{eq-Appendix:II:A=XAX}).
\end{proof}

Note that from Corollary \ref{eqV} it immediately follows that the elements of any two BNT corresponding to the same CF will coincide up to phases and a gauge transformation.

We need a final lemma, whose proof can be found e.g. in \cite{Gemma}.

Let us consider two sets of complex numbers, $\lambda_{\alpha,k_\alpha}=|\lambda_{\alpha,k_\alpha}|e ^{i \phi_{\alpha,k_\alpha}}$ ($\alpha=a,b$, $k_\alpha=1,\ldots,x_\alpha$), sorted such that $|\lambda_{\alpha,k_\alpha}|\ge |\lambda_{\alpha,k_\alpha+1}|$ and
$\phi_{\alpha,k_\alpha}\le \phi_{\alpha,k_\alpha+1}$.

\begin{lem}
\label{Lem:app_simple}
If $\forall N\le \max\{x_a,x_b\}$,
 \be
 \label{eq:app_aux}
 \sum_{k=1}^{x_a} \lambda_{a,k}^N = \sum_{k=1}^{x_b} \lambda_{b,k}^N
 \ee
then  $x_a=x_b$ and $\lambda_{a,k}=\lambda_{b,k}$.
\end{lem}

We can now prove the Fundamental Theorem of MPV, that we restate:

\begin{thm*}
\label{thm1}
Let $A$ and $B$ be two tensors in CF, with BNT $A^i_{k_a}$ and $B^i_{k_b}$ ($k_{a,b}=1,\dots,g_{a,b}$), respectively. If for all $N$, $A$ and $B$ generate MPV that are proportional to each other, then: (i) $g_a=g_b=:g$; (ii) for all $k$ there exists a $j_k$, phases $\phi_{k}$, and non-singular matrices, $X_{k}$ such that $B^i_{k}=e^{i\phi_{k}}X_{k}A^i_{j_k}X_{k}^{-1}$.
\end{thm*}

\begin{cor*}
\label{II_cor2}
If two tensors $A$ and $B$ in CF generate the same MPV for all $N$ then: (i) the dimensions of the matrices $A^i$ and $B^i$ coincide; (ii) there exists an invertible matrix, $X$, such that
 \be
 \label{eq:II_auxcor}
 A^i = X B^i X^{-1}.
 \ee
\end{cor*}

\begin{proof}
To prove the Theorem we reason very similarly as in the proof of Proposition \ref{prop:char-BNT}. Let us first consider $B_k$ for some given $k$. It is not possible that $\langle V^{(N)}(B_k)|V^{(N)}(A_j)\rangle\to 0$ as $N\to\infty$ for all $j$, since otherwise the MPV generated by $A$ and $B$ could not be proportional for all $N$ (Lemma \ref{Lem1}). Thus, according to Corollary \ref{eqV}, there must exist one $j_k$ such that $|V^{(N)}(B_1)\rangle =e^{i\phi_k N} |V^{(N)}(A_{j_k})\rangle$. According to Lemma \ref{equalMPS} we have $B_k = e^{i\phi_k} X_k A_{j_k} X_k^{-1}$. We also conclude that $g_a\ge g_b$, But if we had considered $A_k$ to start with, we would obtain $g_b\ge g_a$, so that $g_a=g_b$.

To prove the Corollary, we denote by $A_j$ and $B_j$ the BNT which, according to the Theorem, have $g$ NT and (after relabeling) $B_j = e^{i\phi_j} Y_j A_{j} Y_j^{-1}$. Denoting by $\mu_{j,q_a}$ and $\nu_{j,q_b}$ ($j=1,\ldots,g$, $q_{x} = 1,\ldots, r_{x,j}$, with $x=a,b$) the coefficients in the decomposition of $A$ and $B$ in their BNT (\ref{eq:II_ABasicTensors}), we must have
 \be
 \sum_{q=1}^{r_{a,j}} \mu_{j,q}^N = \sum_{q=1}^{r_{b,j}} (\nu_{j,q}e^{i\phi_j})^N, \quad \forall j,N,
 \ee
From this equation it follows that $r_{a,j}=r_{b,j}=:r_j$ and $\mu_{j,q}=\nu_{j,q}e^{i\phi_j}$  (Lemma \ref{Lem:app_simple}). Then, the dimensions of the matrices $A^i$ and $B^i$ coincide. Also, defining
 \be
 Y= \oplus_{j=1}^{g} \Id_{r_j}\otimes Y_j
 \ee
and using (\ref{eq:II_ABasicTensors}), we have (\ref{eq:II_auxcor}).
\end{proof}

We finish with the following

\begin{cor}\label{thm:Fundamental-CFII}
If in the above Theorem and Corollary the tensors are in CFII, the same applies but now with the $X$ and $X_k$ unitary.
\end{cor}

\section{Proofs of Section \ref{Sec:MPS}}\label{AppendixPure}

Let us start proving Theorem \ref{thm:renormalization-flow}. For that we start with the observation that the classes of tensors given by the relation (\ref{eq:equivalence-relation-RFP-pure})  are in one to one correspondence with the set of CPM since two sets of matrices give the same CPM via the formula
\be\label{eq:appendix-CPM} \mathcal{E}(X)=\sum_{i=1}^d A^i X A^{i\dagger}
\ee if and only if they are related by an  isometry in the physical indices \cite{Michael-Notes}.

With this identification, it is easy to see that the renormalization step defined by blocking in Section \ref{Sec:MPS} corresponds exactly to $\mathcal{E}\mapsto \mathcal{E}^2$ and, in particular, it does not depend on the representatives of the classes of tensors that are used. Moreover, if the renormalization flow has a limit, it must verify $\mathcal{E}= \mathcal{E}^2$. Using again that two sets of matrices giving the same CPM are related by an isometry in the physical indices we get Theorem \ref{thm:renormalization-flow}. Using Theorem \ref{TheoremZCLPure} (that we will prove below) and the fact that the transfer matrix $\mathbb{E}$ is the matrix associated to the linear map $\mathcal{E}$, this shows also the equivalence $(i)\Leftrightarrow (ii)$ in Theorem  \ref{thm:main-MPS}.

Let us now see that, if we start with a tensor $A$ in CF, the renormalization procedure always converges. We write $A$ in terms of its BNT
(\ref{eq:II_ABasicTensors}). Given the block diagonal form of the matrices $A^i$, the transfer matrix is also block diagonal. In order to simplify the notation, we define
 \be
 \label{eq:II_Etilde}
 \mathbb{E}' = \left(X\otimes \bar X\right) \mathbb{E}
 \left(X\otimes \bar X\right)^{-1},
 \ee
where $X$ is defined in (\ref{eq:II_X}). We have
 \be
 \label{EasEkk}
 \mathbb{E}'= \oplus_{j,j'=1}^g \oplus_{q,q'=1}^{r_j, r_{j'}}\mu_{j,q}{\bar\mu_{j',q'}} \mathbb{E}_{j,j'},
 \ee
where
 \be
 \label{Ekk}
 \mathbb{E}_{j,j'} = \sum_{i=1}^d A^i_j \otimes \bar A^i_{j'}
 \ee
$\mathbb{E}_{j,j}$ can be considered as the matrix associated to the linear CPM (\ref{Ek}). It thus has a single, nondegenerate, eigenvalue of magnitude 1, with corresponding right and left eigenvectors $|R_j)$ and $|L_j)$, i.e.
 \begin{subequations}
 \bea
 \mathbb{E}_{j,j} |R_j)&=&|R_j),\\
 (L_j| \mathbb{E}_{j,j} &=&(L_j|,
 \eea
 \end{subequations}
If the tensor is in CFII, we have in addition
 \begin{subequations}
 \bea
 \label{Lk}
 |L_j) &=& \sum_{\alpha} |\alpha,\alpha), \\
 |R_j) &=& \sum_{\alpha} (\Lambda_j)_{\alpha,\alpha} |\alpha,\alpha),
 \eea
 \end{subequations}
where $\Lambda_j>0$ and diagonal. The eigenvalues of $\mathbb{E}_{j,j'}$ for $j\ne j'$ must have a magnitude $< 1$ (see Lemma \ref{equalMPS}).
The assumed normalization ($|\mu_{jq}|\le 1$ and at least one of them $=1$) implies then that the sequence $\mathbb{E}^N$, and hence the sequence $\mathcal{E}^N$, will converge always.

Let us move now to prove Theorem \ref{TheoremZCLPure}, that is, ZCL (defined as CID and LO) is equivalent to the property $\mathbb{E}^2=\mathbb{E}$ for a tensor $A$ in CF. The above argument shows such a tensor $A$ is LO if and only if $\mathbb{E}_{j,j'}=0$ for all $j\ne j'$. Let's assume then that property and prove that then $\mathbb{E}^2=\mathbb{E}$ is equivalent to CID. This will finish the proof of Theorem \ref{TheoremZCLPure}. Wlog we will assume for simplicity that we work with $\mathbb{E}'$ as defined in (\ref{eq:II_Etilde}).

Clearly, if $\mathbb{E}^2=\mathbb{E}$, equation (\ref{Corr}) shows that $A$ has CID. For the converse we reason as follows. Assume that $\mathbb{E}^2\not =\mathbb{E}$. This implies that there exists a NT in $A$, wlog $A_1$, so that $\mathbb{E}_1^2\not =\mathbb{E}_1$ and hence there is a non-zero eigenvalue $\lambda$ smaller than $1$. Call $|r_1), |l_1)$ to the associated right and left eigenvectors (so that $(l_1|r_1)=1$). The block-injective property guaranteed by Proposition \ref{propblockinj} ensures the existence of operators $O_1, O_2$, supported on regions of size $3 D^5$ so that
\be
\mathbb{E}_{O_1}=P_{j=j'=1}\left( \oplus_{q,q'=1}^{r_1} (\mu_{1,q}\bar\mu_{1,q'})^{3D^5}|R_1)(l_1|\right) P_{j=j'=1}
\ee
and
\be
\mathbb{E}_{O_2}=P_{j=j'=1}\left( \oplus_{q,q'=1}^{r_1} (\mu_{1,q}\bar\mu_{1,q'})^{3D^5} |r_1)(L_1|\right) P_{j=j'=1}
\ee
where $P_{j=j'=1}$ projects on the sector $j=j'=1$ of the decomposition (\ref{EasEkk}).

Due to formula (\ref{Corr}), the correlation function
\be
\langle V^{(N)}(A)|O_1 O_2|V^{(N)}(A)\rangle= \lambda^{{\rm dist}(O_1,O_2)} \sum_{q,q'=1}^{r_1}\left(\mu_{1,q}\bar\mu_{1,q'}\right)^{N}
\ee
Recall that by Lemma \ref{Lem:app_simple},
\be
 \sum_{q,q'=1}^{r_1}\left(\mu_{1,q}\bar\mu_{1,q'}\right)^{N}= \left| \sum_{q=1}^{r_1}\mu_{1,q}^{N}\right|^2\not = 0
\ee
for infinitely many $N$s. Hence the tensor does not have CID, what finishes the proof of Theorem \ref{TheoremZCLPure}.

The next step is to prove Theorem \ref{thm:charact-MPS}, giving the structural characterization for RFP tensors in CF presented in (\ref{III_CFI_RFP}). For that we need  the following lemma.

\begin{lem}\label{lem:charact-NT-pure-RFP}
A NT $A$ is a RFP iff it can be written as
 \be
 \label{eq_III_NT_RFP}
 A^i = X \Lambda U^i X^{-1},
 \ee
where $\Lambda$ is diagonal, positive, with $\tr(\Lambda)=1$, and $U$ is an isometry
 \be
 \sum_{i=1}^s U^i_{\alpha,\beta} \bar U^i_{\alpha',\beta'} =
 \delta_{\alpha,\alpha'}\delta_{\beta,\beta'}.
 \ee
 We can express graphically (\ref{eq_III_NT_RFP}) as follows
 \be
 {\includegraphics[height=5em]{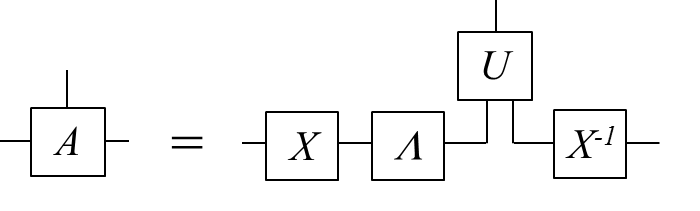}}
 \ee
 \end{lem}

\begin{proof}
Let us take a NT $A$ which is a RFP (the reverse implication is trivial). Let us first assume that $A$ is in CFII.
Since $A$ is a NT, its transfer matrix has a single block.
Moreover, being a RFP, $\mathbb{E}^2=\mathbb{E}$, which implies that $\mathbb{E}=|R)(L|$ with $(L|R)=1$. Moreover,
 \be
 |R)= \sum_{\alpha=1}^D \Lambda_\alpha |\alpha,\alpha),\;\;
 (L|=\sum_{\alpha=1}^D (\alpha,\alpha|
 \ee
This transfer matrix is obtained by a tensor whose physical index is a double-index, $i=(\alpha',\beta')$,  $A^{(\alpha',\beta')}_{\alpha,\beta}= \delta_{\alpha,\alpha'}\delta_{\beta,\beta'} \sqrt{\Lambda_\alpha}$. Any other tensor giving rise to the same transfer matrix must be related to this one by an isometry $U$ acting on the physical indices. Thus, we arrive at (\ref{eq_III_NT_RFP}) with $X=\Id$. The appearance of $X$ follows from the fact that $A$ is in CF, and thus it is related by a similarity transformation to its CFII.
\end{proof}

The proof of Theorem \ref{thm:charact-MPS} follows immediately now  if we use the decomposition of $A$ in BNT elements (\ref{eq:II_ABasicTensors}), and the form of each NT in the RFP given in the previous proposition. Note that the proof of RFP $\Leftrightarrow$ ZCL shown above gives explicitly that if $A$ is a RFP, each one of the elements in the BNT is a RFP too.

Let us finish this section showing the remaining equivalence in Theorem \ref{thm:main-MPS}. Namely that RFP $\Leftrightarrow$ NNCPH.

The implication RFP $\Rightarrow$ NNCPH is trivial from Theorem \ref{thm:charact-MPS}. To prove the reverse, we will use \cite{Beigi} where it is shown that the ground space of any nearest-neighbor commuting Hamiltonian in a 1D spin chain with a finite (independent of system size) degeneracy $g$ is spanned by $g$ states of the form (\ref{eq:basic-vectors-RFP-pure}) that are locally orthogonal. By the defining properties of the parent Hamiltonian, they must be then a BNT for the original tensor $A$, and they are RFP. Tensor $A$ has then the form given in (\ref{III_CFI_RFP}) and is a RFP by Theorem \ref{thm:charact-MPS}.


\section{Proofs of Section \ref{Sec:mixed-states}}
\label{AppendixMixed}

\subsection{Proof of Proposition \ref{PropILILp1}}

\begin{proof}
We show here that for a translationally invariant MPDO of $N$ spins and finite bond dimension, $D$, the mutual information, $I_L$, is a monotonically increasing function ($I_L\le I_{L+1}$) for $L<\lfloor N/2\rfloor$, and that it is bounded by a constant that only depends on $D$. The latter immediately follows from the fact that for any MPDO the mutual information fulfills an area law, $I_L\le 4 \log D$, as proven in \cite{WolfAreaLaw}. In order to prove monotonicity, we divide the spin chain in four consecutive segments, $A,B,C$ and $D$, consisting of $N-2L-1, L, 1$ and $L$ spins, respectively (note that $N-2L-1> 0$ since $L<\lfloor N/2 \rfloor$). We use the strong subaddivity inequality for the entropy, $S_B+S_{ABC}\le S_{AB}+S_{BC}$, where $S_X$ denotes the entropy of the reduce state in region $X$. Substracting $S_{ABCD}=S_N$ on both sides of the inequality, we obtain the desired monotonicity condition.
\end{proof}

\subsection{Proof of Theorem \ref{thm:main-simple}}

In this subsection we will slightly change the graphical representation, due to the fact that in some instances we will have to consider tensors with many indices. We will follow the convention that in tensors generating  MPDO we will represent the bra indices with dashed lines whenever we write them on the other side of the tensors. For instance, the tensor $K$ generating the MPDO $\sigma^{(N)}(K)$ will be interchangeably written as
 \be
 \raisebox{-12pt}{\includegraphics[height=4.5em]{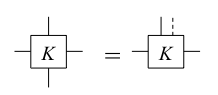}}
 \ee

We start out with the first implication of Theorem \ref{thm:main-simple}, which is the simplest one:

\begin{prop}
\label{propsimple}
If a tensor $K$ is a RFP then it satisfies ZCL and SAL.
\end{prop}

\begin{proof}
Since $K$ is a RFP, there exist two tpCMP, ${\cal T}$ and ${\cal S}$, fulfilling (\ref{TandSforsimple}). The fact that it has ZCL is a direct consequence of (\ref{eq:Tmap}) and that ${\cal T}$ preserves the trace: ${\rm tr}[K_2(X)]={\rm tr}\left[{\cal T}[K_1(X)]\right]={\rm tr}[K_1(X)]$ for all $X$, and thus $K$ has ZCL according to Definition \ref{DefinitionZCL}. In order to prove SAL, let us take the chain,
and remind that according to Proposition \ref{PropILILp1}, for $L<\lfloor N/2\rfloor$ we have that $I_L\le I_{L+1}$. In order to prove the equality, we separate the chain into two regions, A and B, of size $L$ and $N-L$, respectively, with $L<\lfloor N/2\rfloor$, so that the mutual information between A and B is $I_L$. Now we apply ${\cal T}$ to the first spin of region A, and ${\cal S}$ to the first two spins of region B, so that the mutual information between A and B will become $I_{L+1}$. Given that the mutual information cannot increase by local operations, we have that $I_{L+1}\le I_L$, which completes the proof.
\end{proof}

We continue now with a simple property that follows from SAL. Let us define $\sigma_3^{(N)}$ as the state that is obtained from $\sigma^{(N)}(K)$ by tracing all the spins but the first three. Graphically
 \be
\label{sigma3bka}
\sigma_3^{(N)} =
\raisebox{-13.8pt}{\includegraphics[height=4.2em]{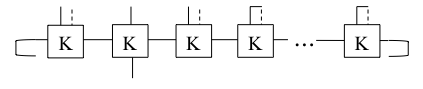}}
 \ee
Let us call those spins $A$, $B$, and $c$, and the corresponding Hilbert spaces, $H_A,H_B$, and $H_c$, respectively. Then we have

\begin{lem}
\label{Lsigma3}
If the tensor $K$ has SAL, then we can decompose $H_B=\oplus_k H_{b_{1}}^{(k)}\otimes H_{b_{2}}^{(k)}$, so that
 \be
 \label{sigma3}
 \sigma_3^{(N)} = \oplus_k \rho_{Ab_1}^{(k)}\otimes \rho_{b_2 c}^{(k)}
 \ee
where $\rho_{Ab_1}^{(k)}$ is supported on $H_A\otimes H_{b_1}^{(k)}$ and $\rho_{b_2 c}^{(k)}$ on  $H_{b_2}^{(k)}\otimes H_c$.
\end{lem}

\begin{proof}
We divide the whole spin chain composed of $N$ spins into four connected regions, $A,B,C$, and $D$, with $1,1,N-3$, and $1$ spin, respectively. SAL implies that $I_2=I_1$, which is equivalent to $S_{ABC}+S_B = S_{AB}+S_{BC}$, i.e. strong subadditivity is fulfilled as an equality. In this case, we can employ the characterization of density operators satisfying this equality \cite{Hay03}: we first divide the Hilbert space $H_{ABC}= \oplus_k H_{Ab_1}^{(k)}\otimes H_{b_2C}^{(k)}$, with $H_{Ab_1}^{(k)}=H_A\otimes H_{b_1}^{(k)}$ ($H_{b_2C}^{(k)}=H_{b_2}^{(k)}\otimes H_C$). Then, we write the reduced state $\sigma_{ABC} = \oplus_k \rho_{Ab_1}^{(k)}\otimes \rho_{b_2C}^{(k)}=\sigma^{(N)}_{N-1}$, where the direct sum is with respect to region $B$, and $\rho_{Ab_1}^{(k)}$ ($\rho_{b_2C}^{(k)}$) is supported on $H_{Ab_{1}}^{(k)}$ ($H_{b_{2}C}^{(k)}$). The operator $\sigma^{(N)}_{N-1}$ is nothing but the original one, $\sigma^{(N)}({\cal K})$, where we have traced the last spin. Now, if we trace all the spins in region $C$ but the first one, we obtain (\ref{sigma3}).
\end{proof}

\begin{defn}
\label{DefQk}
We will denote by $Q_k$ the projector onto $H_{b_{1}}^{(k)}\otimes H_{b_{2}}^{(k)}$.
\end{defn}

Note that in the proof of Lemma \ref{Lsigma3}, the direct sum structure is derived for $\sigma^{(N)}_{N-1}$ and not for $\sigma^{(N)}$. That is, we always trace at least one spin from the chain. This is why in the whole Section \ref{MixedMutual} we have to deal with simple tensors, i.e. that when we trace a spin none of the elements of a BNT are eliminated.

\subsubsection{Case I: $K$ injective}

In this subsection we will assume $K$ to be an injective NT generating MPDOs. In order to emphasize this fact, we will denote the tensor by ${\cal K}$. The fact that it is injective implies that [compare (\ref{injectivity})]
 \be
 \label{iCFK}
 \raisebox{-12pt}{\includegraphics[height=5em]{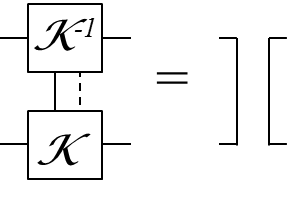}}
 \ee
which will be used in some of the proofs. In the next subsection we will relax this condition and consider general simple tensors.

One of the main purposes of this subsection is to relate SAL and ZCL with the existence of an isometry, $U$, and tensors $r_k,l_k$, such that [compare Proposition \ref{thm:main-simple}]
\be
 \label{AppUkU=rl}
  \raisebox{-12pt}{\includegraphics[height=8em]{MPDO_K=sum_k_LkRk.png}}
\ee
with
\be
 \label{Appetakhetc}
  \eta_{k,h}=\raisebox{-12pt}{\includegraphics[height=3em]{MPDO_etahk=RhLk.png}}\ge 0
\ee
and
 \begin{equation}
 \label{Apptralktrrk}
T_{k,h}:={\rm tr}(\eta_{k,h})= a_kb_h  
\end{equation}
with
 \begin{equation}
 \label{AppPsiPhi}
 \sum_k a_k b_k=1.
 \end{equation}

Let us recall that a non-negative matrix --i.e. a matrix with non-negative elements-- is called primitive \cite{refstochmat} if there exist some finite $n_0$ such that $\langle k|T^{n_0}|h\rangle >0$ for all $k,h$. For a primitive matrix, $\lim_{n\to\infty} (T^n)_{k,h}=a_k b_h$ for some real numbers $a_k$ and $b_h$.

\begin{lem}
\label{propSN}
If ${\cal K}$ fulfills SAL, there exists an isometry, $U$, and tensors $r_k,l_k$, such that (\ref{AppUkU=rl}), with (\ref{Appetakhetc}), and the non-negative matrix $
 T_{k,h}={\rm tr}(\eta_{k,h})$
is primitive.
\end{lem}

\begin{proof}
Since ${\cal K}$ fulfills SAL, we can use the characterization of Lemma \ref{Lsigma3}.
Now, we project the second spin using $Q_k$, and apply the map ${\cal K}^{-1}$ to the first and third spin of $\sigma^{(N)}_3$. We denote by
 \be
 X_{\alpha,\beta}=\raisebox{-13.2pt}{\includegraphics[height=3.2em]{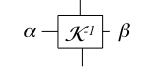}}
 \ee
and
 \be
 \label{Qketc}
 Q_k {\rm tr}_{1,3}\left[(X_{\alpha_1,\beta_1}\otimes \Id_2 \otimes X_{\alpha_3,\beta_3})\sigma_3^{(N)}\right] Q_k = A^{(k)}_{\alpha_1,\beta_1}\otimes B^{(k)}_{\alpha_3,\beta_3}
 \ee
where $A^{(k)}_{\alpha_1,\beta_1}={\rm tr}_1(X_{\alpha_1,\beta_1} \rho_{Ab_1}^{(k)})$ and similarly with $B^{(k)}_{\alpha_3,\beta_3}$. Taking (\ref{sigma3bka}) for $\sigma^{(N)}_3$, and using (\ref{iCFK}) in the
computation of the lhs of (\ref{Qketc}) we get
 \be
 A^{(k)}_{\alpha_1,\beta_1}\otimes B^{(k)}_{\alpha_3,\beta_3}= m_{\beta_3,\alpha_3} \raisebox{-17.2pt}{\includegraphics[height=4.2em]{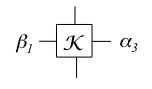}}
 \ee
with
 \be
 m_{\beta,\alpha}=
 \raisebox{-14pt}{\includegraphics[height=3.0em]{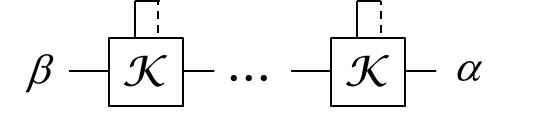}}
 \ee
It is clear that we can always choose $\alpha$ and $\beta$ such that $m_{\beta,\alpha}\ne 0$ [since otherwise $\sigma^{N-3}({\cal K})$ would vanish]. Thus we can write
 \be
 \label{formK}
 \raisebox{-12pt}{\includegraphics[height=9em]{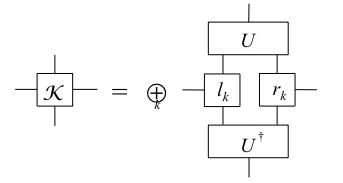}}
 \ee
for some tensors $r_k$ and $l_k$. The direct sum structure refers to the physical (vertical) indices, as it is inherited from the splitting of the space $H_B$ defined above. The isometry $U$ just changes the local basis such that in each subspace defined by $Q_k$, we have the tensor product structure $H_{b_1}^{(k)}\otimes H_{b_2}^{(k)}$. In the following, we will work in that basis, by defining $\tilde \sigma= U^{\dagger\otimes  N} \sigma^{(N)}({\cal K}) U^{\otimes  N}\ge 0$.

We proceed by defining
 \be
 \label{etarl}
 \eta_{k,h}=\raisebox{-22pt}{\includegraphics[height=4.7em]{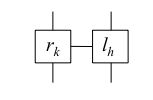}}
 \ee
Then, according to (\ref{formK}) we have
 \be
 0\le \left[Q_{k_1}\otimes\ldots Q_{k_N}\right] \tilde\sigma \left[Q_{k_1}\otimes\ldots Q_{k_N}\right] = \otimes_{n=1}^N \eta_{k_n,k_{n+1}}
 \ee
with $k_{N+1}=k_1$. It is then clear that the $\eta$'s can be chosen positive semidefinite.

We define now
 \be
\label{StochT}
 T_{k,h}={\rm tr}(\eta_{k,h})\ge 0,
 \ee
and thus it is a non-negative matrix. We can always find some finite integer $n$ such that we can write $T^n=\oplus_x T^{(x)}$, where $T^{(x)}$ is primitive for each $x$ \cite{refstochmat} (since by blocking we can get rid of the $p$-periodic components, as it is in the case of CPMs; see discussion after Eq.(\ref{Ek})). For simplicity, we will assume that $n=1$, but the proof can be easily generalized to an arbitrary $n$. Thus, we can split the different $k$ in disjoint sets, $S_x$, so that $\eta_{k,k'}=0$ if $k$ is in a different set than $k'$. Thus, we can write
 \be
 \tilde \sigma = \oplus_{x} \tilde\sigma_x
 \ee
where each
 \be
 \tilde\sigma_x = \oplus_{k_1,\ldots,k_N\in S_x} \otimes_{n=1}^N \eta_{k_n,k_{n+1}}.
 \ee
The tensors generating each $\tilde\sigma_x$ must be locally orthogonal (see Definition \ref{DefLO}), since they correspond to different $k$'s. This is incompatible with the fact that ${\cal K}$ is injective, unless there is a single $x$, i.e. $T$ is primitive. This last statement can be proven as follows: assume that there are two values of $x=1,2$. Then we can find two orthogonal projectors, $P_{1,2}$ with $P_1^{(1)} \tilde\sigma P_2^{(1)}= P_2^{(1)} \tilde\sigma P_1^{(1)}=\left(P_1^{(1)}\otimes P_2^{(2)}\right) \tilde\sigma \left(P_1^{(1)}\otimes P_2^{(2)}\right)=0$, where the superscript indicates on which spin they are acting on. Since ${\cal K}$ is injective, we can use (\ref{iCFK}) to obtain
 \be
 \label{P1KP2}
 \raisebox{-35pt}{\includegraphics[height=8.5em]{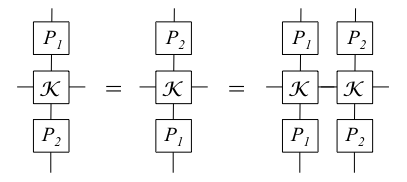}}\quad =0
 \ee
It is now relatively easy to show that this is incompatible with the fact that ${\cal K}$ is injective. For that, let us denote by ${\cal A}_{1,2}$ the linear spaces of matrices generated by ${\rm tr}(X P_{1,2} {\cal K} P_{1,2})$, respectively, for all $X$. According to the first two conditions in (\ref{P1KP2}), the linear space of matrices generated by ${\rm tr}(X{\cal K})$, ${\cal A}$, is ${\cal A}_1+ {\cal A}_2$. Furthermore, the last condition in (\ref{P1KP2}) implies that ${\cal A}_1{\cal A}_2=0$, and thus ${\cal A}$ cannot be the whole set of matrices, from which we conclude that ${\cal K}$ cannot be injective. This statement also holds if there are more than two values of $x$.
\end{proof}

So far we have not introduced the condition of ZCL. It will restrict the properties of the tensors $r_k$ and $l_k$ defined in (\ref{formK}). For the next proposition we need to define
 \be
 \label{lkrk}
 \raisebox{-12pt}{\includegraphics[height=7em]{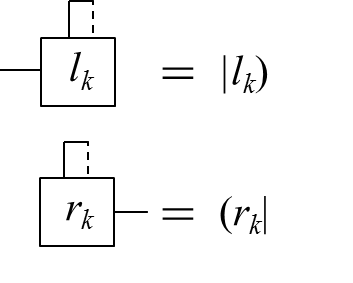}}
 \ee
where the tensors are the ones used to build up the $\eta$'s in (\ref{etarl}). Note that
 \be
 \label{Tkn}
 T_{kh}=(r_k|l_h).
 \ee

\begin{lem}
\label{SALZCL}
If ${\cal K}$ fulfills both SAL and ZCL, then there exist real numbers, $a_k,b_k$, such that (\ref{Apptralktrrk}) and (\ref{AppPsiPhi}).
\end{lem}

\begin{proof}
According to Definition \ref{DefinitionZCL} and (\ref{Tkn}), ZCL implies
 \be
 \sum_k |l_k)(r_k| = \sum_{k,h} T_{k,h} |l_k)(r_h|.
 \ee
Thus,
 \be
 {\rm tr}(T^N)={\rm tr}(T)=\sum_k (r_k|l_k)
 \ee
for all $N$. Using the theory of fixed points of non-negative matrices \cite{refstochmat}, and taking into account that $T$ is primitive (see Lemma \ref{propSN}), this in turn implies that $T$ has one eigenvalue equal to one and the rest are zero, so that $T_{k,h} = a_k b_h$, where $a_k,b_h$ are real numbers. By normalization ${\rm tr}(T)=1$, which results in (\ref{AppPsiPhi}). 
\end{proof}

Putting together lemmas \ref{propSN} and \ref{SALZCL}, we have

\begin{cor}
\label{CorollarytoProp}
If ${\cal K}$ fulfills both SAL and ZCL, there exists an isometry, $U$, tensors $r_k$ and $l_k$, and real numbers $a_k,b_k$ such that (\ref{AppUkU=rl}), with (\ref{Appetakhetc}), (\ref{Apptralktrrk}), and (\ref{AppPsiPhi}).
\end{cor}

Now we prove the reverse implication. 

\begin{prop}
\label{3to5}
If there exists an isometry, $U$, tensors $r_k$ and $l_k$, and real numbers $a_k,b_k$ such that (\ref{AppUkU=rl}), with (\ref{Appetakhetc}), (\ref{Apptralktrrk}), and (\ref{AppPsiPhi}), then there exist two tpCPM, ${\cal T}$ and ${\cal S}$ so that
 \be
 \label{3to5tcpm}
 \raisebox{-12pt}{\includegraphics[height=4.6em]{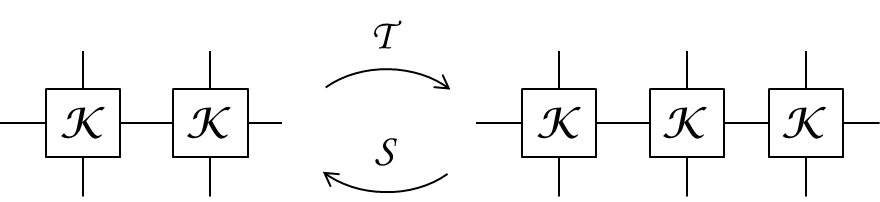}}
 \ee
\end{prop}

\begin{proof}
We will ignore the isometry $U$ from now on, as it does not add anything to the problem.
Given the direct sum structure in (\ref{AppUkU=rl}), in order to be more concise, for each subspace labeled by $k$ we will separate the left and right spaces and call them subspins. For instance, if we have three spins, 1, 2, and 3, we will talk about the first left subspin, 1l, or the second right subspin, 2r.
We define the map ${\cal T}$ as the composition of three tpCPM, ${\cal T}={\cal T}_2{\cal T}_1{\cal T}_0$. The first one, ${\cal T}_0$, transforms two spins into a single one, whereas ${\cal T}_1$ transforms one into three. In each orthogonal subspace, $k$, the first performs a trace in subspins 1r and 2l, and it is thus explicitly a tpCPM. The second one, ${\cal T}_1$, introduces two operators $\eta$. More explicitly,
 \be
 {\cal T}_1(X) = \sum_{k,h} {\cal T}_{k,h} \left[ (Q^l_k\otimes Q^r_h) X (Q_k^l\otimes Q_h^r)\right]
 \ee
Here the $Q_k^{l,r}$ are projectors on the left and right subspaces
corresponding to the index $k$, and
 \be
 {\cal T}_{k,h}(X) = \frac{1}{a_k b_h} X\otimes \left[\oplus_l\left( \eta_{k,l}\otimes \eta_{l,h}\right)\right].
 \ee
This is also a tpCPM, the trace-preserving property being a consequence of
 \be
 \frac{1}{a_kb_h} \sum_{l} {\rm tr}(\eta_{k,l}){\rm tr}(\eta_{l,h})=1
 \ee
where we have used that ${\rm tr}(\eta_{k,h})=a_k b_h$ and (\ref{PsiPhi}). Finally, ${\cal T}_2$ just shifts in each subspace labeled by the different $k$'s, some of the subspins. More explicitly,
 \be
 {\cal T}_2 (X_{k_1}^{1l}\otimes X_{h_1}^{1r}
 X_{k_2}^{2l}\otimes X_{h_2}^{2r}
 X_{k_3}^{3l}\otimes X_{h_3}^{3r}) = X_{k_1}^{1l}\otimes X_{k_2}^{1r}\otimes X_{h_2}^{2l}
 X_{k_3}^{2r}\otimes X_{h_3}^{3l}\otimes X_{h_1}^{3r}.
 \ee
In summary, we have
 \be
 \raisebox{-12pt}{\includegraphics[height=20em]{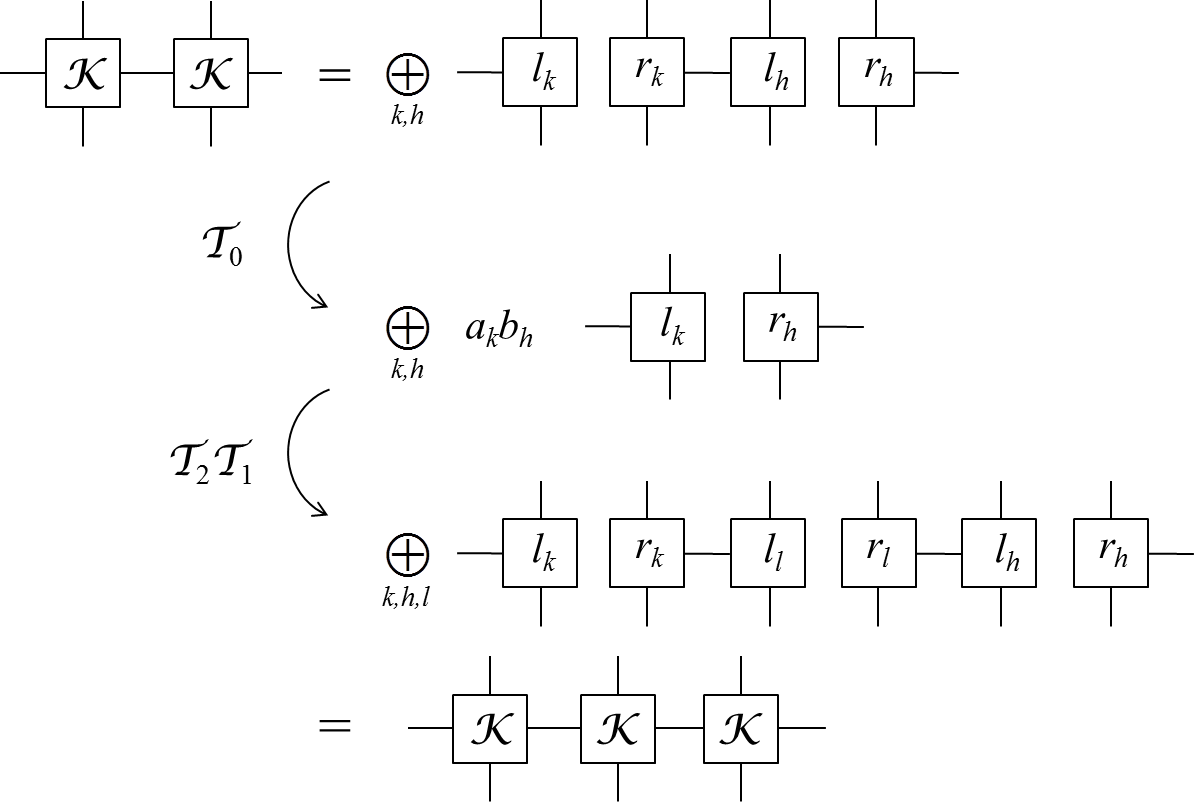}}
 \ee

In a similar way, we build ${\cal S}$ out of three tpCPM, ${\cal S}={\cal S}_2{\cal S}_1{\cal S}_0$. The first one, ${\cal S}_0$ in each subspace just takes the trace of the subspins $1r,2l,2r$ and $2l$. The second one,
 \be
 {\cal S}_1 (X)= \sum_{k,h} {\cal S}_{k,h} \left[ (Q^l_k\otimes Q^r_h) X (Q_k^l\otimes Q_h^r)\right]
 \ee
where
 \be
 {\cal S}_{k,h}(X) = \frac{1}{a_kb_h} X \otimes \eta_{k,h},
 \ee
so that, again, ${\cal S}_1$ is a tpCPM. The last one ${\cal S}_2$ performs a shift similar to ${\cal T}_2$,
 \be
 {\cal S}_2 (X_{k_1}^{1l}\otimes X_{h_1}^{1r}
 X_{k_2}^{2l}\otimes X_{h_2}^{2r}) = X_{k_1}^{1l}\otimes X_{k_2}^{1r}\otimes X_{h_2}^{2l}
 X_{h_1}^{2r}.
 \ee
and so it is tpCPM. We have
 \be
 \raisebox{-12pt}{\includegraphics[height=20em]{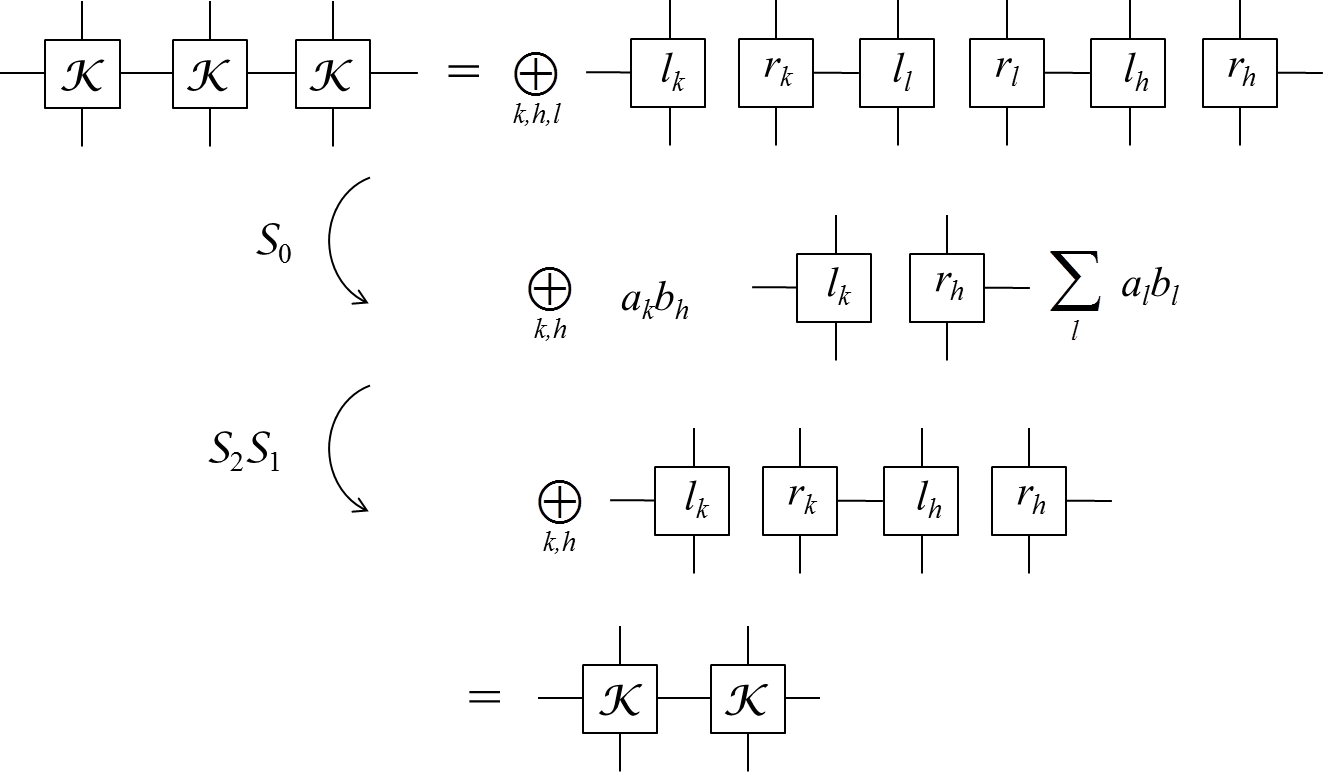}}
 \ee
\end{proof}

We now show the equivalence between GSNNCG + ZCL and SAL + ZCL. 

\begin{prop}
\label{3to4}
The MPDO generated by a ${\cal K}$ fulfilling SAL has the form
 \be
 \label{expression}
 \sigma^{(N)}({\cal K})\propto \prod_{n=1}^N B_{n,n+1},
 \ee
where $B_{n,n+1}\ge 0$ act on nearest neighbors and commute, i.e. $[B_{n-1,n},B_{n,n+1}]=0$.
\end{prop}

\begin{proof}
Since ${\cal K}$ fulfills SAL, Lemma \ref{propSN} implies that 
 \be
 \tilde \sigma = U^{\dagger\otimes  N} \sigma^{(N)}({\cal K})U^{\otimes  N}
 \ee
 has the form
 \be
 \label{sigmaNK2}
 \tilde \sigma= \oplus_{k_1,\ldots,k_N} \left[\otimes_{n=1}^N \eta_{k_n,k_{n+1}} \right]\;.
\ee
Here, the physical Hilbert space of each spin is decomposed as $H_n=\oplus_{k} H_{L,n}^{(k)}\otimes H_{R,n}^{(k)}$ and $\eta_{k_n,k_{n+1}}\ge 0$ is acting on $H_{R,n}^{(k_n)}\otimes H_{L,n+1}^{(k_{n+1})}$.Therefore, we can write 
 \be
 \tilde\sigma= \prod_{n=1}^N B_{n,n+1}
 \ee
where $B_{n,n+1}=\sum_{k_n, k_{n+1}} \Id_{H_{L,n}^{(k_n)}} \otimes \eta_{k_n,k_{n+1}}\otimes \Id_{H_{R,n}^{(k_{n+1})}}$ acts nontrivially only on the $n$-th and $n+1$-th spins and is considered to be identity in the rest. It is  straightforward that the $B$ commute with each other.
\end{proof}

Conversely:
\begin{prop}
\label{4to2}
If ${\cal K}$ generates MPDOs of the form (\ref{expression}) and it has ZCL,
then it fulfills SAL.
\end{prop}

\begin{proof}
It is proven in \cite{Beigi} that form (\ref{expression}) implies form
 (\ref{sigmaNK2}).
 Let us divide the chain into four regions, A, B, C, and D, containing $m-1,1,N-m-1$, and $1$ subsequent spins, respectively, and by $H_X$ ($X=A,B,C,D$) the corresponding Hilbert spaces. The aim is to show that the reduced state in ABC can be written as 
\be
\label{eq:gibbs-sal}
 \oplus_k \rho_{A b_1}^{(k)} \otimes
 \rho_{b_2 C}^{(k)}
 \ee
since this implies saturation of the strong subadditivity inequality between regions A, B, and C \cite{Hay03}, i.e. $S_{ABC}+S_B=S_{AB}+S_{BC}$, which immediately implies $I_1=I_m$, and thus SAL.
Using ZCL tracing $1$ spin (as in subsystem $D$) is the same as taking a chain of $N+1$ spins and tracing two of them. Arguing as in Lemmas  \ref{propSN} and \ref{SALZCL}, we know that $T_{kh}={\rm tr} (\eta_{kh})=a_kb_h$. 
Therefore
$$
\tilde{\sigma}_{ABC}=\oplus_{k_1,\ldots,k_{N-1}} \left[\otimes_{n=1}^{N-2} \eta_{k_n,k_{n+1}}\otimes \sum_{k_N}a_{k_N} \eta_{k_N-1,k_{N}}\otimes \sum_{k_{N+1}}b_{k_{N+1}} \eta_{k_{N+1},k_{1}} \right]
$$
which implies (\ref{eq:gibbs-sal}).
\end{proof}

\subsubsection{Case II: $K$ simple}

Now we move to the more general case. The idea is to show that this case can be reduced to the injective one. Thus, from now on we will consider a tensor $K$, generating MPDO, and that is both simple and in biCF. We will use calligraphic letters, ${\cal K}$, to denote elements of a BNT. In case there are several elements, ${\cal K}_j$, we will use indistinguishably
 \be
 \label{KBNTs}
 \raisebox{-12pt}{\includegraphics[height=5em]{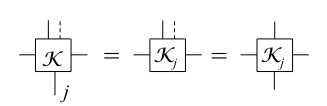}}
 \ee

The strategy to prove Theorem \ref{thm:main-simple} in general will be to first show that the elements in a BNT of $K$ fulfill
 \be
 \label{AppKxKy=0}
  \raisebox{-24pt}{\includegraphics[height=5.5em]{MPDO_KxKy.png}}=0 \quad {\text if} \; x\ne y,
\ee
and then that each of them generates MPDO and fulfills SAL and ZCL, so that we can use the results proven in the previous section. In particular, we will relate those properties to GSNNCH (see Definition \ref{defrhoNComm}):
 \be
 \label{ApprhoNComm}
 \sigma^{(N)}\propto \oplus_{x} n_x e^{-\sum_{j=1}^N \tau_{j}(h^{(x)})}.
 \ee
where $n_x$ are natural numbers and the $h^{(x)}$ acts on the first two spins, $\tau_j$ translates the spins by an amount $j$, and $[h^{(x)},\tau_1(h^{(x)})]=0$.

Let us start out by showing that ZCL implies that the $\mu_{j,q}$ do not depend on $q$.
\begin{lem}
\label{lemmus}
A simple tensor, $K$, generating a MPDO and fulfilling ZCL, must have for all $q$, $\mu_{j,q}=\mu_j$ in its CF (\ref{Eq19}).
\end{lem}

\begin{proof}
We just have to apply the Definition \ref{DefinitionZCL} and use (\ref{Eq19}), so that
\be
 \raisebox{-12pt}{\includegraphics[height=3.6em]{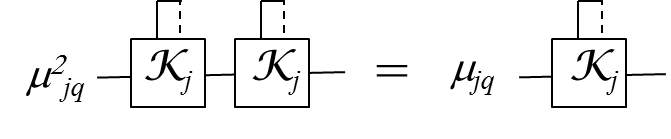}}
 \ee
since this cannot be zero, we conclude the proof.
\end{proof}

According to this lemma, we can include the $\mu_j$ in the definition of the BNT, ${\cal K}_j$, and
explicitly represent this tensor in CF as follows [compare (\ref{Eq19})]:
 \be
 \label{CFK}
 \raisebox{-12pt}{\includegraphics[height=5em]{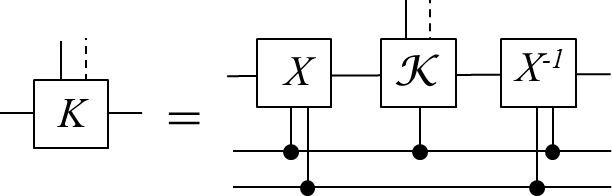}}
 \ee
Since $K$ is in biCF (\ref{defnbi}), there exists another tensor, $K^{-1}$, such that
 \be
 \label{biCFK}
 \raisebox{-12pt}{\includegraphics[height=7.5em]{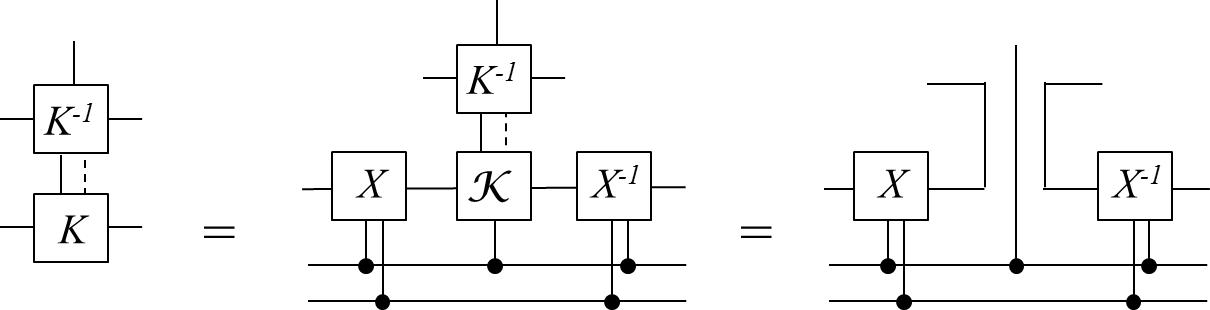}}
 \ee
Note that this implies that
 \be
 \label{Kidentity}
 \raisebox{-12pt}{\includegraphics[height=9em]{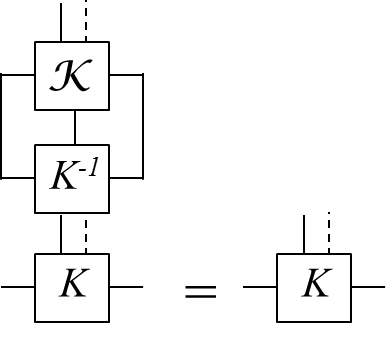}}
 \ee
so that we recover the original tensor after applying $K^{-1}$.

We will need first some intermediate lemmas:

\begin{lem}
If $K$ fulfills SAL, there exist orthogonal projectors, $P_j$ such that
 \be
 \label{Pis}
 \sum_i P_i = \Id,
 \ee
and
 \be
 \label{PjKiPj}
 \raisebox{-40pt}{\includegraphics[height=8em]{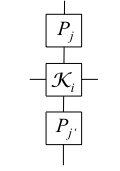}}=0 \quad {\rm if } \; j\ne j'\; {\rm or }\; i\ne j
 \ee
\end{lem}

\begin{proof}
As $K$ fulfills SAL, we can use Lemma \ref{Lsigma3}, and thus $\sigma_3^{(N)}$ must have the form (\ref{sigma3}). Thus, there exist some projectors, $Q_k$ (see Definition \ref{DefQk}) that project onto each of the subspaces labeled by $k$, fulfilling
 \be
 \sum_k Q_k=\Id.
 \ee
We will now show that
 \be
 \label{Qks}
 \raisebox{-40pt}{\includegraphics[height=8em]{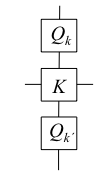}}=0 \;{\rm if } \; k\ne k'
 \ee
and that for each $k$ there exists a unique $j$, such that
 \be
 \label{QkKjs}
 \raisebox{-40pt}{\includegraphics[height=8em]{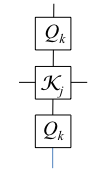}}\ne 0
 \ee
We then define $S_j$ as the set of $k$'s such that (\ref{QkKjs}) holds. As a consequence that for each $k$ there exists a unique $j$ fulfilling (\ref{QkKjs}), these sets are disjoint, and their union is the whole set of $k$'s. Thus, defining
 \be
 P_j = \sum_{k\in S_j} Q_k
 \ee
we will have that they are orthogonal projectors fulfilling (\ref{Pis}). Furthermore, (\ref{PjKiPj}) will follow immediately. Thus, from now on we just have to prove (\ref{Qks}) and (\ref{QkKjs}).

Using the fact that $K$ is simple, for each $j$ we can always find values of the indices, $\alpha_j,\beta_j$, such that
 \be
  \tilde m_j=   \raisebox{-12pt}{\includegraphics[height=3em]{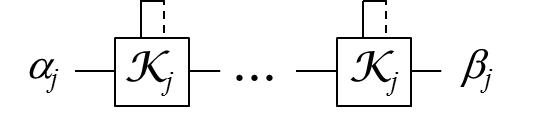}}\ne 0
 \ee
We can also find some $N$ such that $q_j=\sum_q \mu_{jq}^N\ne 0$. We will consider that $N$ from now on, and denote by $m_j=q_j \tilde m_j\ne 0$. Following a similar path as in Lemma \ref{propSN}, we define
 \be
 X_{j,\alpha,\beta}=\raisebox{-17pt}{\includegraphics[height=4.5em]{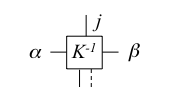}}
 \ee
and apply them to the first and third spin of $\sigma_3^{(N)}$, together with the projectors $Q_k$ and $Q_{k'}$ to the second one, obtaining
 \be
 \label{Qketc2}
  \delta_{k,k'} A^{(k)}_{j,\alpha}\otimes B^{(k)}_{j',\beta}=Q_k {\rm tr}_{1,2}\left[(X_{j,\beta_{j},\alpha}\otimes \Id_2 \otimes X_{j',\beta,\alpha_{j'}})\sigma_3^{(N)}\right] Q_{k'}
 \ee
where $A^{(k)}_{j,\alpha}={\rm tr}_1(X_{j,\beta_j,\alpha} \rho_{Ab_1}^{(k)})$ and similarly with $B^{(k)}_{j',\beta}$. Taking (\ref{sigma3bka}) (with $K$ instead of ${\cal K}$) for $\sigma^{(N)}_3$, and using (\ref{biCFK}) in the
computation of the rhs of (\ref{Qketc2}) we get
 \be
 \label{keyformula}
 \delta_{k,k'} A^{(k)}_{j,\alpha}\otimes B^{(k)}_{j',\beta}= \delta_{j,j'} m_{j} \raisebox{-40pt}{\includegraphics[height=8.0em]{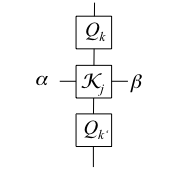}}
 \ee
This is the main equation from which we can derive all the desired results. First, taking $j=j'$, if $k\ne k'$ for all $\alpha,\beta$, the lhs vanishes, so that we must have (\ref{Qks}). Let us denote by
 \be
 O_{j}^{(k)} = \raisebox{-40pt}{\includegraphics[height=8.0em]{MPDO_QkKjQk.png}}
 \ee
Now, for all $k=k'$ there must be some $j=j'$, $\alpha$ and $\beta$ such that the rhs is non-zero, since otherwise the whole $Q_k\sigma^{(N)}_3 Q_k$ would vanish, which wlog can be assumed to be different from zero. Thus, for each $k$ there exist at least one $j_k$ such that $O_{j_k}^{(k)}$ does not vanish. Let us now show that there cannot be another $j\ne j_k$ with the same property. Since $O_{j_k}^{(k)}$ does not vanish, using  (\ref{keyformula}) we conclude that there is some $\alpha=\alpha_k$ such that $A^{(k)}_{j_k,\alpha_k}\ne 0$. Using again (\ref{keyformula}) we conclude that for any $j\ne j_k$ and $\beta$, $B^{(k)}_{j,\beta}=0$. Thus, $A^{(k)}_{j,\alpha}\otimes B^{(k)}_{j,\beta}=0$ for all $\alpha,\beta$, and from (\ref{keyformula}) we obtain that $O_j^{(k)}=0$ as well, as we claimed, and which concludes the proof.
\end{proof}

\begin{prop}
\label{prop2to3}
If $K$ fulfills SAL and ZCL, then the BNT's fulfill (\ref{AppKxKy=0}),  and for each of them there exists an isometry, $U$, tensors $r_k$ and $l_k$, vectors $\Phi$ and $\Psi$, and real numbers $a_k,b_k$ such that (\ref{AppUkU=rl}), with (\ref{Appetakhetc}), (\ref{Apptralktrrk}), and (\ref{AppPsiPhi}).
\end{prop}

\begin{proof}
As shown in the previous lemma, there exist orthogonal projectors $P_j$ fulfilling (\ref{Pis}) and (\ref{PjKiPj}). From those properties, one can easily show that
 \be
 \label{generateMPDO}
 \raisebox{-12pt}{\includegraphics[height=9em]{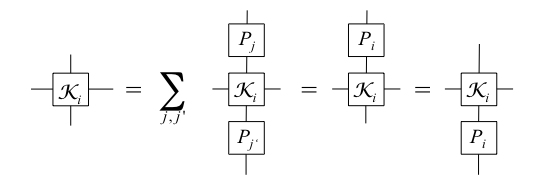}}
 \ee
which immediately implies (\ref{AppKxKy=0}). Next, we will show that each BNT element ${\cal K}_j$ generates MPDO, and fulfills SAL and ZCL. Then, applying Corollary \ref{CorollarytoProp}, we will conclude the proof.

The fact that each  ${\cal K}_j$ generates a MPDO can be shown as follows. Let
us define
 \be
 \label{sigmaNKj}
 \sigma^{(N)}_j(K)= P_j^{\otimes N} \sigma^{(N)}(K) P_j^{\otimes N}= n_j \sigma^{(N)}({\cal K}_j)
 \ee
where we have used (\ref{generateMPDO}) and Lemma \ref{lemmus}, and where $n_j$ is the number of blocks with index $j$. From this identity, we have that  $\sigma^{(N)}({\cal K}_j)\ge 0$.

In order to show that each BNT fulfills SAL, let us divide a chain of $N$ spins into four regions, A,B,C, and D, containing $m-1,1,N-m-1$, and $1$ consecutive spins, respectively. We can write
 \be
 \sigma^{(N)}(K)=\oplus_j P^{(B)}_j \sigma^{(N)}(K) P^{(B)}_j
 =\oplus_j \sigma^{(N)}_j(K)
 \ee
where the projectors are applied to the spin in region B only, and we have used (\ref{generateMPDO}). Let us denote by $p_j={\rm tr}[\sigma^{(N)}_j(K)]\ge 0$, by $S_X^j$ the entropy of $\sigma^{(N)}_j(K)/p_j$ in region $X$, and by $H_p=-\sum_{j} p_j \log p_j$. We will also denote by $S_X$ the entropy of $\sigma^{(N)}(K)$ in region $X$. Given that the different $\sigma^{(N)}_j(K)$ are supported on different subspaces in region $B$, we can write
 \be
 \label{entropiesj}
 S_X = H_p + \sum_j p_j S_X^j.
 \ee
Now, using the fact that $K$ fulfills SAL, and thus $I_1=I_m$ for $m<\lfloor N/2\rfloor$, this implies $S_{ABC}+S_B=S_{AB}+S_{BC}$. Using (\ref{entropiesj}) we obtain
 \be
 \label{sumstron}
 \sum_j p_j (S_{ABC}^j + S_B^j) = \sum_j p_j (S_{AB}^j + S_{BC}^j).
 \ee
Together with the strong subadditivity inequality for each of the $\sigma^{(N)}_j(K)$, i.e.
 \be
 \label{strongj}
 S_{ABC}^j + S_B^j \le S_{AB}^j + S_{BC}^j
 \ee
this implies that for each $j$, the (\ref{strongj}) must be fulfilled as an equality, if one of them would be strictly smaller, then it would be impossible to have (\ref{sumstron}). Substracting $S_N^j$ on each side of the equality we obtain that the mutual information is independent of $m$, and thus that each $\sigma^{(N)}_j(K)/p_j$ saturates the area law. Using the relation between $\sigma^{(N)}_j(K)$ and $\sigma^{(N)}({\cal K}_j)$ given in (\ref{sigmaNKj}), we conclude that each ${\cal K}_j$ fulfills SAL. Furthermore, since $K$ has ZCL then using the CF (\ref{CFK}) it is obvious that each ${\cal K}_j$ must also have ZCL.
\end{proof}

Next, we prove the following

\begin{prop}
\label{prop3to4}
If the BNT's of $K$ fulfill (\ref{AppKxKy=0}), and (\ref{AppUkU=rl}), with (\ref{Appetakhetc}), (\ref{Apptralktrrk}), and (\ref{AppPsiPhi}), then the MPDO generated by $K$ have the form (\ref{ApprhoNComm}).
\end{prop}

\begin{proof}
Given (\ref{AppKxKy=0}) and Lemma \ref{lemmus}, we can write
 \be
 \sigma^{(N)}(K)=\oplus_j n_j \sigma({\cal K}_j),
 \ee
where $n_j$ is an integer. The rest follows directly from Proposition \ref{3to4} since, according to the previous proposition, each ${\cal K}_j$ fulfills SAL and ZCL.
\end{proof}

Now we show

\begin{prop}
\label{prop4to2}
If $K$ generates MPDO of the form (\ref{ApprhoNComm}) then it fulfills SAL.
\end{prop}

\begin{proof}
Using Proposition \ref{4to2} to each of the terms in the direct sum we have that each ${\cal K}_j$ fulfills SAL. Since the direct sum is local, this also applies to the whole $K$.
\end{proof}

\begin{prop}
\label{prop2to5}
If $K$ fulfills SAL and ZCL then there exist two tpCPM ${\cal T}$ and ${\cal S}$ fulfilling (\ref{3to5tcpm}).
\end{prop}

\begin{proof}
The proof follows in the same way as that of Proposition \ref{3to5} by first projecting into each local subspace labeled by $j$.
\end{proof}

Now we are at the position of proving Theorem \ref{thm:main-simple}:

\begin{proof}
$i \Rightarrow ii $ by Proposition \ref{propsimple}.
$ii \Rightarrow iii$ by Proposition \ref{prop2to3}.
$iii \Rightarrow iv$ by Proposition \ref{prop3to4}. $iv \Rightarrow ii$ by Proposition \ref{prop4to2}. $iii \Rightarrow v$ by Proposition \ref{prop2to5} where the tpCPM are ${\cal T}^2$ and ${\cal S}^2$.
\end{proof}

\subsection{Proof of Proposition \ref{Prop:IV.12}.}

In order to prove it, we will need the following Lemma about MPV.

\begin{lem}\label{Lemma-L}
Let $A$ be a tensor in CF generating a family of MPV, with local Hilbert space $\mathcal{H}_d$. Let $Y,Z$ be operators on $\mathcal{H}_d$ If
$Y_1V^{(N)}(A)=Z_1V^{(N)}(A)$ for all $N$ (where $Y_1,Z_1$ means that the operators are located in the first spin) -- graphically
 \be\label{eq1:proof.IV.12}
 \raisebox{-12pt}{\includegraphics[height=5.5em]{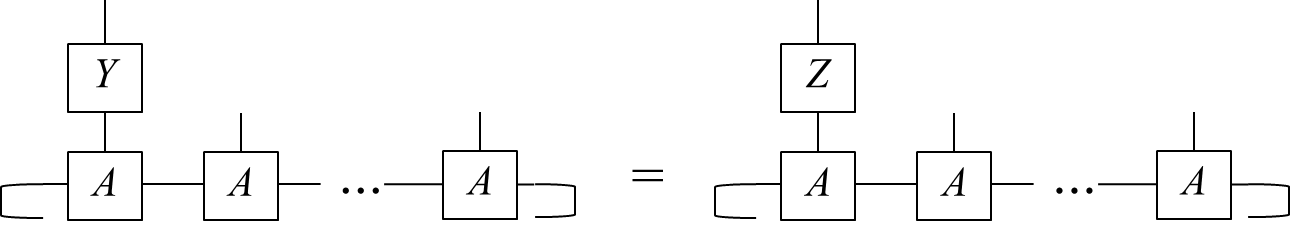}}\nonumber\\
 \ee
then
 \be
 \raisebox{-12pt}{\includegraphics[height=5.5em]{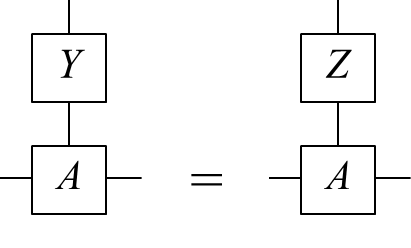}}\nonumber\\
\ee

\end{lem}

\begin{proof}
Let us consider the CF decomposition of $A$ given in (\ref{eq:II_ABasicTensors}).  The hypothesis implies that
 \be
 \raisebox{-12pt}{\includegraphics[height=5.5em]{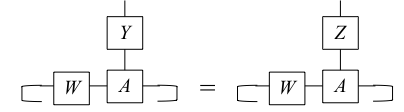}}\nonumber\\
\ee
 for all $W\in {\rm span}\{\oplus_{j,q} \mu_{j,q}X_{j,q}A_jX_{j,q}^{-1}\}^N$. That is, calling
 \be
 \raisebox{-12pt}{\includegraphics[height=12em]{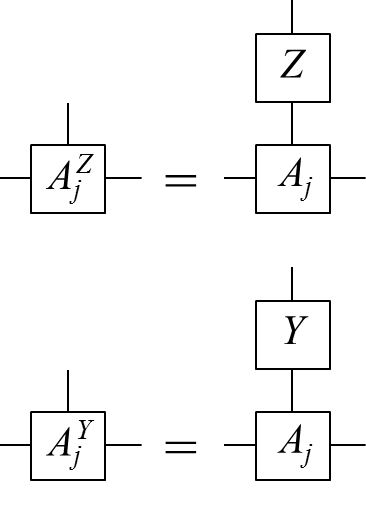}}\nonumber\quad ,\\
\ee
Proposition \ref{propblockinj} ensures that $$\left(\sum_q\mu_{j,q}^N \right)\tr[W_j A_j^Y]=\left(\sum_q\mu_{j,q}^N \right)\tr[W_j A_j^Z]$$ for all $j$, all blocks $W_j$ of $W$, and all $N$. By choosing $N$ so that $\sum_q\mu_{j,q}^N \not = 0$, which exists by Lemma \ref{Lem:app_simple}, we get that $A_j^Y=A_j^Z$ for all $j$ and from there the desired result.
\end{proof}

We can now prove Proposition \ref{Prop:IV.12}.. Let us first restate it.

\begin{prop*}
A tensor $\tilde{M}$ in CF, generating MPDO, is also in CF vertically. Moreover, there exists an isometry $U$ such that
 \be
 \label{UMU-appendix}
 U \tilde{M} U^\dagger = \bigoplus_\alpha \mu_\alpha \otimes M_\alpha,
 \ee
where $\mu_\alpha$ are diagonal and positive matrices, and $\{M_\alpha\}_\alpha$ is a BNT.
\end{prop*}

\begin{proof}
Let us first show that for all projector $P$ so that  $P\tilde{M}=P\tilde{M}P$, we have that $P^\perp\tilde{M}P=0$.

Let us denote by $H^{(N)}=H^{(N)}(\tilde{M})$ the MPDO generated by $\tilde{M}$ in the horizontal direction. That is
 \be
 H^{(N)}(\tilde{M})=\raisebox{-15pt}{\includegraphics[height=4.5em]{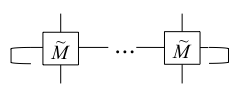}}
 \ee
 Imagine there exists a projector $P$ so that $P\tilde{M}=P\tilde{M}P$. Then
 \bea
 \label{eq1:proof.IV.12}
 PH^{(N)} &=&
 \raisebox{-12pt}{\includegraphics[height=6em]{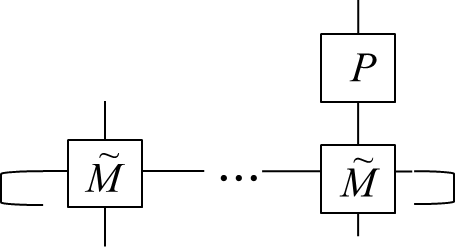}}\nonumber\\
 &=& PH^{(N)}P=(PH^{(N)}P)^\dagger=H^{(N)}P
\eea
By assumption, $\tilde{M}$ is in CF in the horizontal direction. Then, Lemma \ref{Lemma-L} gives $P^\perp\tilde{M}P=0$ as desired.

The next step is to show that there are no p-periodic vectors. If this were the case, there would exist an orthogonal projector $Q$ so that $QH^{(N)}\not = H^{(N)}Q$ for all $N$ but
\be \label{eq2:proof.IV.12} Q\left[H^{(N)}\right]^p =\left[H^{(N)}\right]^pQ \text{ for all } N.
\ee

But since  $H^{(N)}$ is semidefinite positive, (\ref{eq2:proof.IV.12}) implies that $QH^{(N)} = H^{(N)}Q$ for all $N$, which is the desired contradiction.

It only remains to show that there exists an isometry $U$ so that $$U\tilde{M}U^\dagger = \bigoplus_\alpha \mu_\alpha \otimes M_\alpha$$
where the $\mu_\alpha$ are positive diagonal matrices and the $M_\alpha$ a BNT for $\tilde{M}$ in the vertical direction.

Since $\tilde{M}$ is in CF in the vertical direction, it can be written as $\oplus_{\alpha,k}\mu_{\alpha,k} X_{\alpha,k}M_\alpha X_{\alpha, k}^{-1}$. Let us denote by $P_{\alpha,k}$ the projector selecting the sector $\alpha,k$ in such direct sum. Then, since $\tilde{M}$ was also in CF in the horizontal direction, there exists an $N$ so that $0\not =P_{\alpha,k}H^{(N)}P_{\alpha,k}\ge 0$. Its trace is then $\mu_{\alpha,k} d_\alpha>0$, where
 \be
 \raisebox{-12pt}{\includegraphics[height=5.5em]{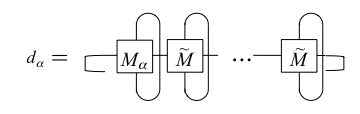}}
 \ee
Since wlog $\mu_{\alpha,1}>0$ for all $\alpha$, we can conclude that $d_\alpha>0$ and hence $\mu_{\alpha,k}>0$. We now finish if we show that
\be\label{eq3:proof.IV.12}
X_{\alpha,k}^\dagger X_{\alpha,k} =\omega_{\alpha,k} X_{\alpha,1}^\dagger X_{\alpha, 1} \text{ for all } k.
\ee
Since we can always choose wlog $X_{\alpha,1}=\Id$ for all $\alpha$,  this would imply that $$\frac{1}{\sqrt{\omega_{\alpha,k}}}X_{\alpha,k}=U_{\alpha,k}.$$
Since the $X_{\alpha,k}$ are defined up to a constant, we are done.

Let us prove then (\ref{eq3:proof.IV.12}). Using that $P_{\alpha,k}H^{(N)}P_{\alpha,k}$ is self-adjoint and that $\mu_{\alpha,k}>0$ we arrive at

 \be
 \raisebox{-12pt}{\includegraphics[height=8.0em]{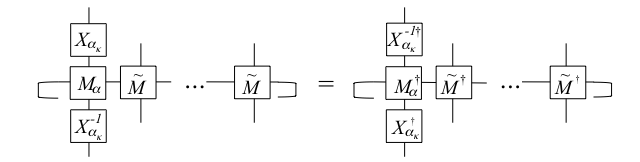}}
 \ee

 Using that $\tilde{M}$ is in CF in the horizontal direction together with Lemma \ref{Lemma-L} we get :

 \be
 \raisebox{-12pt}{\includegraphics[height=14em]{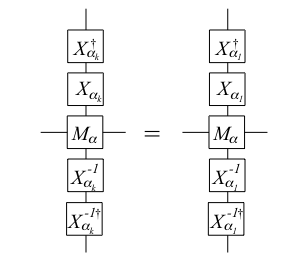}}
 \ee

 But since $M_{\alpha}$ is a NT in the vertical direction, we obtain (\ref{eq3:proof.IV.12}) for a (necessarily positive) constant $\omega_{\alpha,k}$.
 \end{proof}

\subsection{Proof of Theorem \ref{thm:IV.13} }

 Let us restate for convenience here Theorem \ref{thm:IV.13}

 \begin{thm*}
Given a tensor $M$ in CF  that generates MPDO, the following statements are equivalent:
 \begin{itemize}
 \item[(i)] $M$ is a RFP.
 \item[(ii)] There exists a set of diagonal matrices, $\chi_{\alpha,\beta,\gamma}$, with positive elements, such that for each $L$ the operators $O_L(M_\alpha)$ linearly span an algebra with structure coefficients $c_{\alpha,\beta,\gamma}^{(L)}={\rm tr}(\chi^L_{\alpha,\beta,\gamma})$:
 \be
 \label{eq:algebra-appendix}O_L(M_\alpha)O_L(M_\beta)=\sum_\gamma c_{\alpha,\beta,\gamma}^{(L)}O_L(M_\gamma),
 \ee
 and
  \be
 \label{idempotent-appendix}
 m_\gamma = \sum_{\alpha,\beta} c^{(1)}_{\alpha,\beta,\gamma} m_\alpha m_\beta
 \ee
 where $m_\alpha={\rm tr}(\mu_\alpha)$. That is, the vector $(m_\alpha)_\alpha$ is an idempotent for the ``multiplication'' induced by $c^{(1)}$.
 \item[(iii)] There exist isometries, $U_{\alpha,\beta}$, such that
 \be
 U_{\alpha,\beta} M_\alpha M_\beta U_{\alpha,\beta}^\dagger = \oplus_\gamma \chi_{\alpha,\beta,\gamma} \otimes M_\gamma,
 \ee
where $\chi_{\alpha,\beta,\gamma}$ are the same as in the previous statement, and (\ref{idempotent-appendix}) is fulfilled.
 \end{itemize}
\end{thm*}
\begin{proof}
$(i)\Rightarrow (ii) \text{ and } (iii)$.
By Proposition \ref{Prop:IV.12}, both
\be\label{Figure9}
 \raisebox{-12pt}{\includegraphics[height=4em]{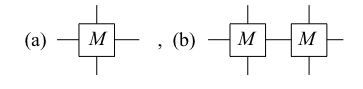}}
\ee
are in CF in vertical direction. So, up to isometries that can be absorbed in the $T$ and $S$ appearing in the RFP property, they can be written as $\oplus_\alpha \mu_\alpha \otimes M_\alpha $ and $\oplus_\beta \nu_\beta \otimes A_\beta$, respectively, where $\{M_\alpha\}_\alpha$  and $\{A_{\beta}\}_\beta$ are BNT with dimensions $d^1_\alpha$ and $d^2_\beta$, respectively, and $\mu_\alpha, \nu_\beta$ are diagonal positive matrices with $m_\alpha=\tr[\mu_\alpha]$ and $n_\beta=\tr[\nu_\beta]$.
Let us define the $C^*$-algebras $\mathcal{A}_1=\oplus_\alpha \mathcal{M}_{d^1_\alpha\times d^1_\alpha}$ and ${\mathcal{A}}_2=\oplus_\beta \mathcal{M}_{{d}^2_\beta\times {d}^2_\beta}$ and the operators
\begin{align}
R_1&:\mathcal{A}_1\longrightarrow \oplus_\alpha {\rm span}\{\mu_\alpha\}\otimes \mathcal{M}_{d^1_\alpha\times d^1_\alpha}, \\
R_2&: \mathcal{A}_2\longrightarrow \oplus_\beta {\rm span}\{\nu_\beta\}\otimes \mathcal{M}_{d^2_\beta\times d^2_\beta},\\
\tilde{R}_1&: \oplus_\alpha {\rm span}\{\mu_\alpha\}\otimes \mathcal{M}_{d^1_\alpha\times d^1_\alpha}\longrightarrow \mathcal{A}_1\\
\tilde{R}_2&:\oplus_\beta {\rm span}\{\nu_\beta\}\otimes \mathcal{M}_{d^2_\beta\times d^2_\beta}\longrightarrow\mathcal{A}_2\
\end{align}
by

\begin{align}
R_1(\oplus_\alpha X_\alpha)&=\oplus_\alpha \frac{\mu_\alpha}{m_\alpha}\otimes X_\alpha\\
R_2(\oplus_\beta Y_\beta)&=\oplus_\beta \frac{\nu_\beta}{n_\beta}\otimes Y_\beta\\
\tilde{R}_1(\oplus_\alpha \lambda_\alpha\mu_\alpha\otimes X_\alpha)&= \oplus_\alpha \lambda_\alpha m_\alpha X_\alpha\\
\tilde{R}_2(\oplus_\beta \delta_\beta\nu_\beta\otimes Y_\beta)&=\oplus_\beta \delta_\beta n_\beta Y_\beta
\end{align}

 Let us call $\tilde{T}=\tilde{R}_2 TR_1:\mathcal{A}_1\rightarrow \mathcal{A}_2$ and $\tilde{S}= \tilde{R}_1SR_2:\mathcal{A}_2\rightarrow \mathcal{A}_1$, where $T,S$ are the tpCPM in the definition of RFP. Let us show that $\tilde{S}\tilde{T}$ is the identity in $\mathcal{A}_1$ and  $\tilde{T}\tilde{S}$  the identity in $\mathcal{A}_2$.
 For that we note that
 $\tilde{T}(\oplus_\alpha m_\alpha M_\alpha(X))=\oplus_\beta n_\beta A_\beta(X)$ and $ \tilde{S}(\oplus_\beta n_\beta A_\beta(X))=\oplus_\alpha m_\alpha M_\alpha(X)$, where
 \be
 \raisebox{-12pt}{\includegraphics[height=3.7em]{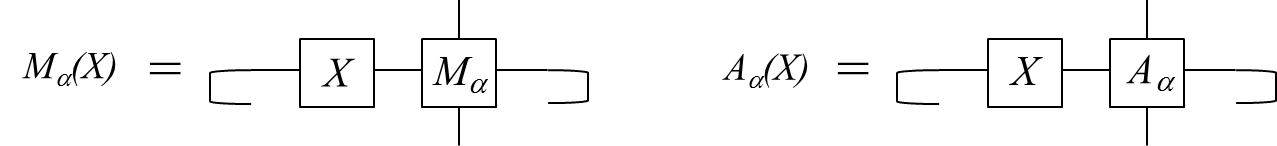}}
\ee
 Hence $\oplus_\alpha m_\alpha M_\alpha(X)$ is contained in the set of fixed points $\mathcal{F}$ of $\tilde{S}\tilde{T}$. By \cite[Theorem 6.14]{Michael-Notes}, $\mathcal{F}$ has the form
 \be
 \mathcal{F}=U\left[0_r\oplus_k\rho_k \otimes \mathcal{M}_{h_k\times h_k}\right]U^\dagger
 \ee
 for a unitary $U$.  Given a vector subspace $V$ of $\mathcal{M}_{D\times D}$ and $L\in \mathbb{N}$, we call
\be
 C_L(V) ={\rm span}(V^L)={\rm span}\{v_1v_2\cdots v_L: v_i\in V\}.
 \ee
 Since $V=\{\oplus_\alpha m_\alpha M_\alpha(X): X\}\subset \mathcal{F}$, we have that $C_L(V)\subset C_L(\mathcal{F})$ for all $L$. But Proposition \ref{propblockinj} guarantees that for some $L$, $C_L(V)=\mathcal{A}_1$. Moreover
 \be
 C_L(\mathcal{F})=U\left[0_r\oplus_k\rho_k^L \otimes \mathcal{M}_{h_k\times h_k}\right]U^\dagger\; .
 \ee

This clearly implies that $\mathcal{F}=\mathcal{A}_1$ and hence  $\tilde{S}\tilde{T}$ is the identity in $\mathcal{A}_1$.

 A similar argument proves that  $\tilde{T}\tilde{S}$  the identity in $\mathcal{A}_2$.

Reasoning as in \cite[Theorem 8]{Inverse-eigenvalue-problem} we get that (after relabeling of blocks that can be absorbed in the definition of $T$ and $S$), $\mathcal{A}_1=\mathcal{A}_2$ and $\tilde{T}(\oplus X_\alpha)=\oplus_\alpha U_\alpha X_\alpha U_\alpha^\dagger$, $\tilde{S}(\oplus Y_\alpha)=\oplus_\alpha U_\alpha^\dagger Y_\alpha U_\alpha$ for some unitaries $U_\alpha$.

Now
\be
\oplus_\alpha \nu_\alpha \otimes A_\alpha(X)=T(\oplus_\alpha \mu_\alpha \otimes M_\alpha(X))=
\ee
\be
=R_2\tilde{T}\tilde{R}_1(\oplus_\alpha \mu_\alpha \otimes M_\alpha(X))=\oplus_\alpha \frac{m_\alpha}{n_\alpha}\nu_\alpha \otimes U_\alpha M_\alpha (X)U_\alpha^\dagger
\ee
from where we conclude (remember we did a relabelling of blocks)
\be
A_\alpha =\frac{m_\alpha}{n_\alpha} U_\alpha M_\alpha U_\alpha^\dagger\;, \; \forall \alpha
\ee

Then, if we denote (\ref{Figure9}.(b)) by $B$, we get that $O_L(B)$ can be written simultaneously as
\be
O_L(B)= \sum_\gamma \frac{m_\gamma^L}{n_\gamma}\tr[\nu_\gamma^L]O_L(M_\gamma)
\ee
and
\be\label{eq1-proof-main-thm}
O_L(B)=\sum_{\alpha,\beta} \tr[\mu^L_\alpha]\tr[\mu_\beta^L] O_L(M_\alpha)O_L(M_\beta)
\ee

From here, it is easy to see that a BNT for

 \be\label{Figure11}
 \raisebox{-12pt}{\includegraphics[height=3.5em]{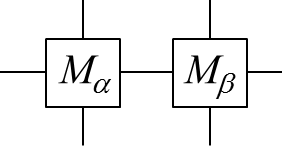}}
\ee
must contain neccessarily only elements of the form $e^{i\theta_{\alpha,\beta,\gamma}} Z_{\alpha,\beta,\gamma} M_\gamma Z_{\alpha,\beta,\gamma}^{-1}$ so that
\be\label{eq2-proof-main-thm}
\text{(\ref{Figure11})}=\oplus_\gamma \chi_{\alpha,\beta,\gamma} \otimes Z_{\alpha,\beta,\gamma} M_\gamma Z_{\alpha,\beta,\gamma}^{-1}
\ee
for a diagonal matrix $\chi_{\alpha,\beta,\gamma}$. This gives (\ref{eq:algebra-appendix}).  Reasoning as in Proposition \ref{Prop:IV.12} we get that $\chi_{\alpha,\beta,\gamma}$ are positive and $Z_{\alpha,\beta,\gamma}$ unitaries.
It only remains to show (\ref{idempotent-appendix}).
By plugging (\ref{eq2-proof-main-thm}) into (\ref{eq1-proof-main-thm}) we get
\be
 \sum_\gamma \frac{m_\gamma^L}{n_\gamma}\tr[\nu_\gamma^L]O_L(M_\gamma) =
 \ee
 \be
 =O_L(B)=\sum_{\alpha,\beta,\gamma} \tr\left[\left(\mu_\alpha\otimes \mu_\beta \otimes \chi_{\alpha,\beta,\gamma}\right)^L\right]O_L(M_\gamma)
\ee
Since the $M_\gamma$ form a BNT, for all $\gamma$ and for all sufficiently large $L$,
\be
\tr\left[\left(\frac{m_\gamma}{n_\gamma}\nu_\gamma\right)^L\right]=\sum_{\alpha,\beta}  \tr\left[\left(\mu_\alpha\otimes \mu_\beta \otimes \chi_{\alpha,\beta,\gamma}\right)^L\right]
\ee
By Lemma \ref{Lem:app_simple}, for all $\gamma$ the set $\{\frac{m_\gamma}{n_\gamma}\nu_{\gamma,k}\}_k$ coincides with the set $\{\mu_{\alpha,k_1}\mu_{\beta,k_2}\chi_{\alpha,\beta,\gamma,k_3}\}_{\alpha,\beta,k_1,k_2,k_3}$ (note that all coefficients are positive). In particular, summing over all indices we get (\ref{idempotent-appendix}).

\

Clearly $(iii)\Rightarrow (ii)$ so it only remains to show that $(ii)\Rightarrow (i)$.

As above, we consider (\ref{Figure9}) and write them, up to conjugation by unitaries, as $\oplus_\alpha \mu_\alpha \otimes M_\alpha $ and $\oplus_\beta \nu_\beta \otimes A_\beta$. (\ref{eq:algebra-appendix}) implies that for all $L$,
 \be
 \sum_\gamma \tr[\nu_\gamma^L]O_L(A_\gamma) =\sum_{\alpha,\beta,\gamma} \tr\left[\left(\mu_\alpha\otimes \mu_\beta \otimes \chi_{\alpha,\beta,\gamma}\right)^L\right]O_L(M_\gamma)
\ee

Since both $A_\gamma$ and $M_\gamma$ are BNT, this implies that, up to relabelling,
\be
A_\gamma = e^{i\theta_\gamma} Z_\gamma M_\gamma Z_\gamma^{-1}
\ee
By Proposition \ref{Prop:IV.12} $e^{i\theta_\gamma}=1$ and $Z_\gamma=U_\gamma$ is unitary.
 Since, for sufficiently large $L$, $O_L(A_\gamma)=O_L(M_\gamma)$ are linearly independent, and all coefficients $\nu_{\gamma,k}$, $\mu_{\alpha,k_1}$, $\mu_{\beta,k_2}$, $\chi_{\alpha,\beta,\gamma,k_3}$ are positive, Lemma \ref{Lem:app_simple} implies that for all $\gamma$ the set $\{\nu_{\gamma,k}\}_k$ coincides with the set $\{\mu_{\alpha,k_1}\mu_{\beta,k_2}\chi_{\alpha,\beta,\gamma,k_3}\}_{\alpha,\beta,k_1,k_2,k_3}$. If we sum in all indices we get that
 \be
 \tr[\nu_\gamma]=\sum_k\nu_{\gamma,k}=\sum_{\alpha,\beta,k_1,k_2,k_3}\mu_{\alpha,k_1}\mu_{\beta,k_2}\chi_{\alpha,\beta,\gamma,k_3}=
 \ee
 \be
 = \sum_{\alpha,\beta} m_\alpha m_\beta c_{\alpha,\beta,\gamma}^{(1)}=m_\gamma
 \ee
 by (\ref{idempotent-appendix}). We can now define $T=\tilde{R}_2\tilde{T}R_1$  where
 \begin{itemize}
 \item $R_1:\mathcal{M}_{d\times d}\longrightarrow \oplus_\gamma\mathcal{M}_{d_\gamma\times d_\gamma}$
 is defined by  $R_1(X)=\sum_\gamma \tr_{\gamma,1}[P_\gamma X P_\gamma]$, being $P_\gamma$ the projector on the corresponding sector of the decomposition $\oplus \mu_\gamma\otimes M_\gamma$ and $\tr_{\gamma,1}$ the partial trace on the first subsystem of that sector.
 \item $\tilde{T}:\oplus_\gamma\mathcal{M}_{d_\gamma\times d_\gamma}\longrightarrow \oplus_\gamma\mathcal{M}_{d_\gamma\times d_\gamma}$ is just conjugation by $\oplus_\gamma U_\gamma$ and
 \item $\tilde{R}_2: \oplus_\gamma\mathcal{M}_{d_\gamma\times d_\gamma}\longrightarrow \mathcal{M}_{d\times d}\otimes \mathcal{M}_{d\times d}$ is given by
 $\tilde{R}_2(\oplus_\gamma X_\gamma)= \oplus_\gamma \frac{\nu_\gamma}{m_\gamma} \otimes X_\gamma$.
 \end{itemize}

 It is clear that $T(\oplus_\gamma \mu_\gamma\otimes M_\gamma)=\oplus_\gamma \nu_\gamma \otimes A_\gamma$ as desired in the definition of RFP.

 Similarly we can define $S=\tilde{R}_1\tilde{S} R_2$ where
 \begin{itemize}
 \item $R_2:\mathcal{M}_{d\times d}\otimes \mathcal{M}_{d\times d} \longrightarrow \oplus_\gamma\mathcal{M}_{d_\gamma\times d_\gamma}$
 is defined by  $R_2(X)=\sum_\gamma \tr_{\gamma,1}[Q_\gamma X Q_\gamma]$, being $Q_\gamma$ the projector on the corresponding sector of the decomposition $\oplus \nu_\gamma\otimes A_\gamma$ and $\tr_{\gamma,1}$ the partial trace on the first subsystem of that sector.
 \item $\tilde{S}:\oplus_\gamma\mathcal{M}_{d_\gamma\times d_\gamma}\longrightarrow \oplus_\gamma\mathcal{M}_{d_\gamma\times d_\gamma}$ is just conjugation by $\oplus_\gamma U_\gamma^\dagger$ and
 \item $\tilde{R}_1: \oplus_\gamma\mathcal{M}_{d_\gamma\times d_\gamma}\longrightarrow \mathcal{M}_{d\times d}$ is defined
  by
 $\tilde{R}_1(\oplus_\gamma X_\gamma)= \oplus_\gamma \frac{\mu_\gamma}{m_\gamma} \otimes X_\gamma$.
 \end{itemize}
 It is clear that $S$ has the property  $S(\oplus_\gamma \nu_\gamma A_\gamma)=\oplus_\gamma \mu_\gamma M_\gamma$ as desired in the definition of RFP.
This finishes the proof of Theorem \ref{thm:IV.13}.

 \end{proof}

\section{Additional results}\label{AppendixLooseEnds}

In this Appendix we elaborate on some of the definitions and results obtained in this paper. In particular, we give alternative definitions of RFP, and exploit the relation between commuting Hamiltonians and the fact that certain correlation functions vanish.

\subsection{Alternative definitions of RFP}

There are in principle other possible definitions of RFP for a tensor. The results of this paper show that the chosen ones seem to be the right ones, given their clear connection with relevant physical notions such as zero-correlation length, saturation of the area law, commuting Hamiltonians or topological properties of boundary theories. It is however worth mentioning other possibilities and their limitations.

The first possibility would be to allow in Definition \ref{RFPMixedTS} that the maps $T$ and $S$ introduce a gauge transformation in the virtual level of the MPDO.
 \be\label{RFP-gauge}
 \raisebox{-20pt}{\includegraphics[height=9em]{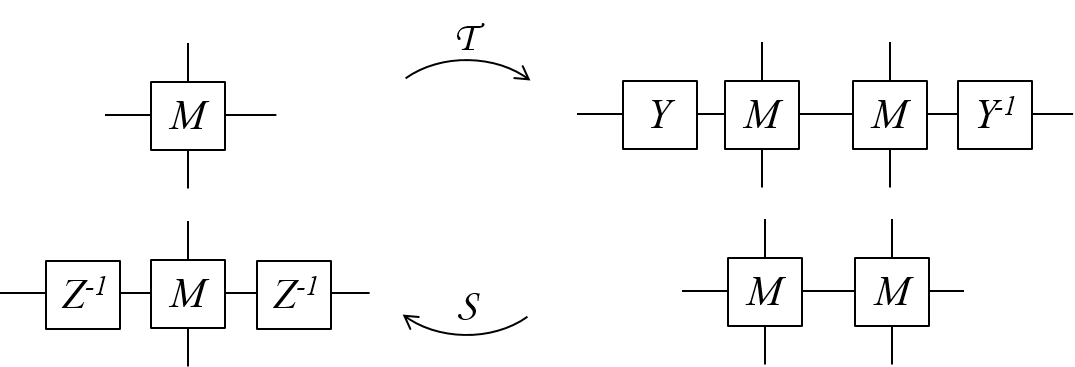}}\\
\ee

In the pure case, where $T$ and $S$ are unitary conjugations,  Definition (\ref{RFP-gauge}) coincides with Definition \ref{RFPMixedTS} for tensors in CF as the following argument shows. (\ref{RFP-gauge}) implies that $\mathbb{E}^2$ is similar to $\mathbb{E}$ and hence both must have eigenvalues $1$ and $0$ exclusively. By going to CFII, one may assume wlog that $\mathbb{E}$ is contractive for the trace norm and hence a standard argument \cite{Fan92} shows that the eigenvalue $1$ cannot have non-trivial Jordan blocks. Using that $\mathbb{E}^{k}$ is similar to $\mathbb{E}$ for all $k$ one gets that the eigenvalue $0$ cannot have non-trivial Jordan blocks either. As a consequence $\mathbb{E}^2=\mathbb{E}$ and, by the characterization given in Theorem \ref{thm:main-MPS}, the tensor is an RFP.

Unfortunately we do not know whether for mixed states a similar argument can be applied and  Definition (\ref{RFP-gauge}) coincides with Definition \ref{RFPMixedTS}. We leave it as an interesting open question.

Another possibility could be to impose that $T$ and $S$ are given by unitary conjugations also in the mixed case. Graphically,
 \be\label{Strong-RFP}
 \raisebox{-20pt}{\includegraphics[height=10em]{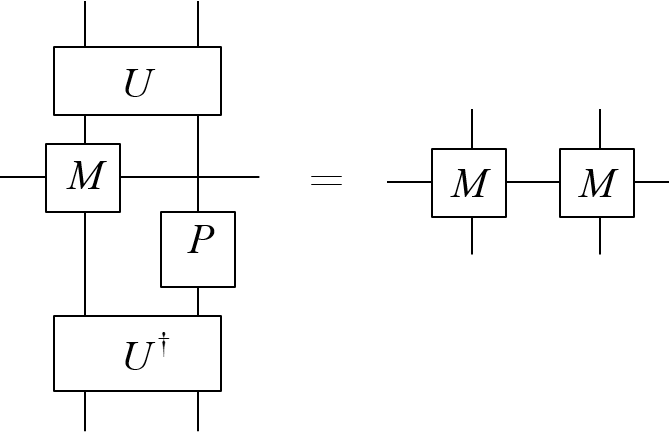}}\\
\ee
for some $P\ge 0$.

This definition is, however, too strong and many of the interesting examples arising from string-net models would not be covered. The reason is that, with (\ref{Strong-RFP}), the rank of the associated MPDO of $N$ sites must grow with $N$ as $rs^{N-1}$ for natural numbers $r,s$. Boundary theories of many string-net models, such as the Fibonacci model, cannot have such behavior. As an illustration, let us see it for one of the elements in the BNT of the boundary theory associated to the Fibonacci model. Concretely, the one corresponding to the vacuum \cite{FrankStrings}. The tensor $A$ has bond dimension $2$ and three physical qubits that we call $ijk$. For each qubit, the ket index is equal to the bra index (the tensor $A$ generates MPDO which are diagonal in the computational basis). Therefore, $A$ gets defined by the elements $A^{ijk}_{\alpha,\beta}$ and the generated MPDO are given by $\rho^{(N)}(A)=$
\be
\sum_{i_1,j_1,k_1,\ldots, i_N,j_N,k_N}\tr(A^{i_1,j_1,k_1}\cdots A^{i_N,j_N,k_N}) \ket{i_1j_1k_1\cdots i_Nj_Nk_N}\bra{i_1j_1k_1\cdots i_Nj_Nk_N}\; .
\ee
All we will need here is that
$A^{ijk}_{\alpha,\beta}=\delta_{i,\alpha}\delta_{k,\beta}N_{ijk}$ and $N_{ijk} =0$ if and only if exactly two of the $i,j,k$ are $0$; otherwise $N_{ijk}>0$.  (This corresponds to the fusion rules of the Fibonacci model.) Let us see that
\be\label{rank-Fibonacci}
{\rm rank}(\rho^{(N)}(A))=\tau_+^{2N}+\tau_-^{2N}.
\ee
where $\tau_\pm=\frac{1\pm\sqrt{5}}{2}$. It is trivial to check that $\tau_+^{2N}+\tau_-^{2N}$ is not of the form $rs^{N-1}$.

Denote by $x^n_{\alpha\beta}$ the rank of the MPDO of length $n$ generated by $A$ with left and right boundary conditions given by $\alpha,\beta\in \{0,1\}$ respectively:
\be
 \sum_{i_1,j_1,k_1,\ldots, i_N,j_N,k_N}(\alpha|A^{i_1,j_1,k_1}\cdots A^{i_N,j_N,k_N}|\beta)\ket{i_1j_1k_1\cdots i_Nj_Nk_N}\bra{i_1j_1k_1\cdots i_Nj_Nk_N}\; .
\ee
Clearly, the $x^n$ satisfy the following recurrence relations:
\begin{align*}
x^n_{00}&=x^{n-1}_{00}+x^{n-1}_{01}\\
x^n_{01}&=x^{n-1}_{00}+2x^{n-1}_{01}\\
x^n_{10}&=x^{n-1}_{10}+x^{n-1}_{11}\\
x^n_{11}&=x^{n-1}_{10}+2x^{n-1}_{11}\; .
\end{align*}
Let us call $x^n$ to the vector $(x^n_{00},x^n_{01},x^n_{10},x^n_{11})$.
We have
$x^1=(1,1,1,2), x^n= x^{n-1}A=x^1 A^{n-1}$ with $$A=\left(\begin{matrix}
1 & 1& 0& 0\\
1 & 2 & 0 & 0\\
0& 0& 1& 1\\
0& 0& 1& 2
\end{matrix}\right)\; .$$
Diagonalizing $$B=\left(\begin{matrix}
1& 1\\
1 & 2
\end{matrix}\right)$$ we obtain $B=PDP^{-1}$ with $P^t=P^{-1}$ and
$$D=\left(\begin{matrix}
\tau_+^2 & 0& \\
0 & \tau_-^2
\end{matrix}\right)\; , \quad P= \left(\begin{matrix}
\frac{1}{\sqrt{1+\tau_+^2}} & \frac{1}{\sqrt{1+\tau_-^2}}\\
\frac{\tau_+}{\sqrt{1+\tau_+^2}} & \frac{\tau_-}{\sqrt{1+\tau_-^2}}
\end{matrix}\right)\; .$$
Since we are considering periodic boundary conditions, the rank of $\rho^{(N)}$ is exactly $x^n_{00}+x^n_{11}=$
\begin{equation}\label{eq:1} x^1 \left(\begin{matrix} P & 0\\
0 & P
\end{matrix}\right)
\left(\begin{matrix}
D^{n-1} & 0 \\
0 & D^{n-1}
\end{matrix}
\right)
\left(\begin{matrix} P^t & 0\\
0 & P^t
\end{matrix}\right)
\left(\begin{matrix}
1\\
0\\
0\\
1
\end{matrix}\right)\; .
\end{equation}
Using the defining property $\tau_\pm^2=\tau_\pm +1$ it is an straightforward computation to verify that (\ref{eq:1}) is exactly $$\tau_+^{2N}+\tau_-^{2N}\; .$$

\subsection{General relations between NNPCH and correlation functions}

In Section \ref{Sec:MPS} (see Theorem \ref{thm:main-MPS}) we have shown that for a spin chain, ZCL of a tensor is equivalent to the existence of a NNPCH for the generated states. In this appendix we show that there is a deeper relation between those concepts which is independent of the dimension or the fact that states are generated by tensor networks.

Let us consider a spin system in any dimension and two disjoint regions, $A$ and $B$. We are interested in analyzing the correlations between operators acting on those regions for certain states. We will denote by $X$ the spins not included in $A$ or $B$, and by $H_{A,X,B}$, the corresponding Hilbert spaces. We will also denote by $H_{AX}=H_A\otimes H_X$, and by $H_{XB}=H_X\otimes H_B$. Let us consider $K_{AXB}$, a  subspace of the whole Hilbert space and denote by $P_{AXB}$ the projector onto that subspace.

\begin{defn}
Given a subspace $K_{AXB}$ we say that regions A and B are decorrelated if for all observables $O_{A,B}$ acting on $H_{A,B}$, respectively,
 \be
 P_{AXB} O_A P_{AXB}^\perp O_B P_{AXB}=0.
 \ee
\end{defn}

Note that we can equivalently write
 \be
 \label{decorr}
 \langle \varphi|O_A P_{AXB}^\perp O_B|\varphi'\rangle=\langle \varphi|O_B P_{AXB}^\perp O_A|\varphi'\rangle=0
 \ee
for all $\varphi,\varphi'\in K_{AXB}$. In particular, if the subspace has just one vector, $\varphi$, we will have
 \be
 \langle O_A O_B\rangle=\langle O_A\rangle \langle O_B \rangle
 \ee
where the expectation value is taken with respect to $\varphi$. This definition of correlations has been used, for example, in \cite{Hastingsdefcorr}. 

Our aim is to connect this concept of correlations with the fact that $K_{AXB}$ is the ground space of certain Hamiltonian, $H$, whenever $A$ and $B$ are not neighbors. We take $H$ local in the sense that it can be written as $H=Q_{AX} + Q_{XB}$, where $Q_{AX}$ ($Q_{XB}$) acts on $H_{AX}$ ($H_{XB}$). We will consider frustration-free problems, so that we can take wlog $Q_{AX}$ and $Q_{XB}$ to be projectors. We will denote by $K_{AX}$ and $K_{XB}$ the kernels of $Q_{AX}$ and $Q_{XB}$, respectively. 
\begin{defn}
Given a subspace $K_{AXB}$ we say that it corresponds to a parent commuting Hamiltonian if there exist two projectors, $Q_{AX}$ and $Q_{XB}$, acting on $H_{AX}$ and $H_{XB}$ respectively, such that
 \begin{subequations}
 \bea
 \label{QAXQXB}
 [Q_{AX},Q_{XB}]&=& 0,\\
 \label{KAXBKAXKXB}
 K_{AXB} &=& \left(K_{AX}\otimes H_B \right)\cap\left( H_A\otimes K_{XB}\right)
 \eea
 \end{subequations}
\end{defn}

Note that the last property (\ref{KAXBKAXKXB}) means that $K_{AXB}$ coincides with the ground subspace of $H=Q_{AX}+Q_{XB}$.
In the following we will show that the two definitions given above are equivalent.

\begin{prop}
Given the subspace $K_{AXB}$, it corresponds to a parent commuting Hamiltonian if and only if regions A and B are decorrelated.
\end{prop}

\begin{proof}
{\em (if):} Let us denote by $K_{AX}$ the support of tr$_B(K_{AXB})$ and by $K_{XB}$ the support of tr$_A(K_{AXB})$, and by $P_{AX}$ and $P_{XB}$ the corresponding projectors. Obviously, if $\varphi\in K_{AXB}$ we will have
 \be
 \label{varphiPAXB}
 |\varphi\rangle = P_{AX}|\varphi\rangle = P_{XB} |\varphi\rangle
 \ee
so that 
 \be
 \label{Propproj}
 P_{\alpha}P_{AXB}=P_{AXB} P_{\alpha} = P_{AXB}.
 \ee
where $\alpha=AX,XB$. Now, we choose orthonormal bases $\{a_i\}\in H_A$, $\{b_i\}\in H_B$, $\{\alpha_i\}\in H_{AX}$ and $\{\beta_i\}\in H_{XB}$. We can always find some vectors $\varphi_{a,i},\varphi_{b,i}\in K_{AXB}$, $\tilde a_i\in H_A$ and $\tilde b_i\in H_B$ such that
 \begin{subequations}
 \bea
 |\alpha_i\rangle&=&\langle\tilde b_i|\varphi_{b,i}\rangle,\\
 |\beta_i\rangle&=&\langle\tilde a_i|\varphi_{a,i}\rangle.
 \eea
 \end{subequations}
We can thus write
 \begin{subequations}
 \label{PAXPXB}
 \bea
 P_{AX}\otimes\Id_B &=& \sum_{i,j} \langle\tilde b_i|\varphi_{b,i}\rangle\langle\varphi_{b,i}|\tilde b_i\rangle\otimes |b_j\rangle\langle b_j|\nonumber\\
  &=& \sum_{i,j} |b_j\rangle_B \langle\tilde b_i|\varphi_{b,i}\rangle\langle\varphi_{b,i}|\tilde b_i\rangle_B \langle b_j|,\\
 \Id_A \otimes P_{XB} &=& \sum_{i',j'} |a_{j'}\rangle\langle a_{j'}|\otimes \langle\tilde a_{i'}|\varphi_{a,i'}\rangle\langle\varphi_{a,i'}|\tilde a_{i'}\rangle\nonumber\\
&=& \sum_{i',j'} |a_{j'}\rangle_A \langle\tilde a_{i'}|\varphi_{a,i'}\rangle\langle\varphi_{a,i'}|\tilde a_{i'}\rangle_A\langle a_{j'}|
 \eea
 \end{subequations}
where we have used the subindices $A,B$ in order to emphasize that the operators act on the corresponding regions. Using these expression we have
 \bea
 \label{PXBAXetc}
 P_{XB}P_{AX} &=& P_{XB}P_{AXB}P_{AX} = P_{AXB},\nonumber\\
 P_{AX} P_{XB} &=& P_{AX}P_{AXB}P_{XB} = P_{AXB}.
 \eea
where we have used (\ref{Propproj}) as well as (\ref{decorr}) with $\varphi=\varphi_{b,i}$, $\varphi'=\varphi_{a,i'}$, $O_A=|a_{j'}\rangle\langle \tilde a_{i'}|$ and $O_B=|\tilde b_i\rangle\langle b_j|$. We thus arrive at $P_{AX}P_{XB}=P_{XB}P_{AX}$, which implies (\ref{QAXQXB}) if we define $Q_{AX}$ and $Q_{XB}$ through 
 \begin{subequations}
 \label{PPPAXB}
 \bea
 \label{PAXQAX}
 P_{AX}&=& \Id - Q_{AX},\\
 \label{PXBQXB}
 P_{XB}&=& \Id - Q_{XB}.
 \eea
 \end{subequations}
 Furthermore, from (\ref{varphiPAXB}) it immediately follows that 
 \be K_{AXB} \subseteq \left(K_{AX}\otimes H_B \right)\cap\left( H_A\otimes K_{XB}\right).
 \ee
  The converse relation can be proven as follows. Let us take any $\varphi\in K_{AX}\otimes H_B$ and $\varphi\in H_A\otimes K_{XB}$. Thus,
 \be
 |\varphi\rangle = P_{AX} |\varphi\rangle= P_{BX} |\varphi\rangle= 
 P_{XB} P_{AX} |\varphi\rangle= P_{AXB} |\varphi\rangle,
\ee
where we have used (\ref{PXBAXetc}). Thus, $\varphi\in K_{AXB}$. 

{\em (only if):} Let us define (\ref{PAXQAX},\ref{PXBQXB}), and $P=P_{AX}P_{XB}$. Given that $K_{AXB} \subseteq \left(K_{AX}\otimes H_B \right)\cap\left( H_A\otimes K_{XB}\right)$ we have that if $\varphi\in K_{AXB}$, then (\ref{varphiPAXB}) and thus (\ref{Propproj}). We have to show that
 \be
 \label{arrggg}
 P=P_{AX}P_{XB}=P_{AXB}
 \ee
In such case, for any $O_{A,B}$ acting on $H_{A,B}$, respectively, we have:
 \bea
 P_{AXB} O_A P_{AXB}^\perp O_B P_{AXB} &=& P_{AXB} P_{XB} O_A P_{AXB}^\perp O_B P_{AX} P_{AXB}\nonumber\\ 
&=& P_{AXB} O_A P_{XB} P_{AXB}^\perp P_{AX} O_B P_{AXB}=0\qquad
 \eea
since $ P_{XB} P_{AXB}^\perp P_{AX}= P_{XB} (\Id -P_{AXB}) P_{AX}= P_{AXB}-P_{AXB}=0$, where we have used (\ref{arrggg}) and (\ref{Propproj}). In order to show (\ref{arrggg}), we recall that since $[P_{AX},P_{XB}]=0$, $P$ must be a projector itself, i.e. $P=P^\dagger=P^2$. Furthermore, if $P_{AXB}|\varphi\rangle = |\varphi\rangle$, then we have (\ref{varphiPAXB}) and thus $P|\varphi\rangle=|\varphi\rangle$. Conversely, if $P|\varphi\rangle=|\varphi\rangle$ then $Q_{AX}|\varphi\rangle=Q_{XB}|\varphi\rangle=0$, so that $\varphi\in \left(K_{AX}\otimes H_B \right)\cap\left( H_A\otimes K_{XB}\right)$. But then, according to (\ref{KAXBKAXKXB}) $\varphi\in K_{AXB}$ so that $P_{AXB}|\varphi\rangle=|\varphi\rangle$.
\end{proof}


\end{document}